\documentclass[conference]{IEEEtran}
\IEEEoverridecommandlockouts
\usepackage{cite}
\usepackage{amsmath,amssymb,amsfonts}
\usepackage{graphicx}
\usepackage{textcomp}
\def\BibTeX{{\rm B\kern-.05em{\sc i\kern-.025em b}\kern-.08em
    T\kern-.1667em\lower.7ex\hbox{E}\kern-.125emX}}

\usepackage[utf8]{inputenc}
\usepackage{mathtools}
\usepackage{amsthm,caption,enumitem,subcaption,nicefrac, amsmath}
\usepackage{fullpage}
\usepackage{dsfont}
\usepackage{bm}
\usepackage{algorithm}
\usepackage[noend]{algorithmic}
\usepackage[usenames,dvipsnames,table]{xcolor}
\usepackage{multirow}
\usepackage{caption}
\usepackage{subcaption}
\usepackage{hyperref}
\hypersetup{
	colorlinks=true,
	linkcolor=Blue,
	citecolor=black!30!OliveGreen,
	filecolor=Maroon,
	urlcolor=Maroon
}
\usepackage[capitalize,noabbrev,nameinlink]{cleveref}
\usepackage{algorithm}
\usepackage[noend]{algorithmic}
\usepackage{tikz}
\usetikzlibrary{positioning,arrows.meta}

\setlist[itemize]{label={$\triangleright$}}

\newtheorem{theorem}{Theorem}

\newtheorem{lemma}[theorem]{Lemma}
\newtheorem{corollary}[theorem]{Corollary}
\newtheorem{fact}[theorem]{Fact}
\newtheorem{assumption}[theorem]{Assumption}
\crefname{assumption}{Assumption}{Assumptions}

\newtheorem{question}[theorem]{Question}
\theoremstyle{definition}
\newtheorem{definition}[theorem]{Definition}

\newcommand{\bM}{{\bm M}}
\newcommand{\bP}{{\bm P}}
\newcommand{\bQ}{{\bm Q}}
\newcommand{\bbH}{\mathbb{H}}
\newcommand{\cC}{\mathcal{C}}

\newcommand{\cX}{\mathcal{X}}
\newcommand{\cY}{\mathcal{Y}}
\newcommand{\cM}{\mathcal{M}}
\newcommand{\cE}{\mathcal{E}}
\newcommand{\cR}{\mathcal{R}}

\newcommand{\cA}{\mathcal{A}}

\newcommand{\eps}{\varepsilon}
\newcommand{\N}{\mathbb{N}}
\newcommand{\Z}{\mathbb{Z}}
\newcommand{\R}{\mathbb{R}}

\DeclareMathOperator{\E}{\mathbb{E}}
\newcommand{\cD}{\mathcal{D}}

\newcommand{\td}{\tilde{d}}
\newcommand{\tx}{\tilde{x}}

\newcommand{\tr}{\tilde{r}}

\newcommand{\bu}{{\bm u}}
\newcommand{\bx}{{\bm x}}

\newcommand{\balpha}{{\boldsymbol{\alpha}}}
\newcommand{\beps}{{\boldsymbol{\varepsilon}}}
\newcommand{\bdelta}{{\boldsymbol{\delta}}}

\renewcommand{\phi}{\varphi}

\newcommand{\hC}{\widehat{C}}
\newcommand{\hM}{\widehat{M}}
\newcommand{\hT}{\widehat{T}}
\newcommand{\hL}{\widehat{L}}
\newcommand{\hV}{\widehat{V}}
\newcommand{\hP}{\widehat{P}}

\newcommand{\oD}{\overline{D}}
\newcommand{\ou}{\overline{u}}
\newcommand{\oM}{\overline{M}}
\newcommand{\obM}{\overline{\bM}}
\newcommand{\obu}{\overline{\bu}}

\DeclareMathOperator{\Lap}{Lap}

\DeclareMathOperator{\negl}{negl}
\DeclareMathOperator{\poly}{poly}
\DeclareMathOperator{\inds}{ind}
\DeclareMathOperator{\diam}{diam}
\newcommand{\indicator}{\mathds{1}}
\newcommand{\NP}{\mathsf{NP}}

\newcommand{\set}[1]{\left \{ #1 \right \}}
\newcommand{\bit}{\{0,1\}}

\newcommand{\inparen}[1]{\left ( #1 \right )}

\newcommand{\seq}[1]{\set{#1}}  %

\definecolor{Gred}{RGB}{219, 50, 54}
\definecolor{Ggreen}{RGB}{60, 186, 84}
\definecolor{Gblue}{RGB}{72, 133, 237}
\definecolor{Gyellow}{RGB}{247, 178, 16}
\definecolor{ToCgreen}{RGB}{0, 128, 0}
\definecolor{myGold}{RGB}{231,141,20}
\definecolor{myBlue}{rgb}{0.19,0.41,.65}
\definecolor{myPurple}{RGB}{175,0,124}

\allowdisplaybreaks

\newcommand{\diO}{\mathsf{di}\mathcal{O}}
\newcommand{\pcS}{\mathsf{pcS}}%
\newcommand{\diOpcS}{\diO\text{-}\mathsf{for}\text{-}\pcS}
\newcommand{\plaindiO}{\mathsf{plain}\text{-}\diO}
\newcommand{\pcdiO}{\mathsf{pc}\text{-}\diO}
\newcommand{\diOgenS}{\diO\text{-}\mathsf{for}\text{-}\mathsf{genS}}
\newcommand{\gendiO}{\mathsf{gen}\text{-}\diO}
\newcommand{\iO}{\mathsf{i}\mathcal{O}}
\newcommand{\aux}{\mathsf{aux}}
\newcommand{\Sampler}{\mathsf{Sampler}}
\newcommand{\LDSSampler}{\mathsf{LDS}\text{-}\mathsf{Sampler}}
\newcommand{\CDP}{\mathsf{CDP}}
\newcommand{\SDP}{\mathsf{SDP}}

\newcommand{\RR}{\mathsf{RR}}
\newcommand{\LDS}{\mathsf{LDS}}
\newcommand{\eval}{\mathsf{eval}}
\newcommand{\VLDS}{\mathsf{VLDS}}
\newcommand{\NBP}{\mathsf{NBP}}

\newcommand{\Mdio}{\cM_{\mathrm{diO}}}
\newcommand{\Mdioaux}{\cM_{\mathrm{diO}}^{\mathrm{aux}}}
\newcommand{\Mcdp}{\cM_{\mathrm{cdp}}}
\newcommand{\Mtuning}{\cM_{\mathrm{tuning}}}
\newcommand{\Mbase}{\cM_{\mathrm{base}}}
\newcommand{\ADI}{\mathcal{A}^{\mathrm{DI}}}
\newcommand{\ACRH}{\mathcal{A}^{\mathrm{CRH}}}
\newcommand{\Bin}{\mathrm{Bin}}

\newcommand{\myparagraph}[1]{\noindent {\em #1}}

\begin{document}

\title{
	Towards Separating Computational and Statistical Differential Privacy 
}

\author{
\IEEEauthorblockN{Badih Ghazi}
\IEEEauthorblockA{\textit{Google Research} \\
badihghazi@gmail.com}
\and
\IEEEauthorblockN{Rahul Ilango}
\IEEEauthorblockA{\textit{MIT$^*$\thanks{$^*$Part of the work done during an internship at Google Research.}} \\
rilango@mit.edu}
\and
\IEEEauthorblockN{Pritish Kamath}
\IEEEauthorblockA{\textit{Google Research} \\
pritish@alum.mit.edu}
\and
\IEEEauthorblockN{Ravi Kumar}
\IEEEauthorblockA{\textit{Google Research} \\
ravi.k53@gmail.com}
\and
\IEEEauthorblockN{Pasin Manurangsi}
\IEEEauthorblockA{\textit{Google Research} \\
pasin@google.com}
}

\maketitle

\begin{abstract}\boldmath
Computational differential privacy (CDP) is  a natural relaxation of the standard notion of (statistical) differential privacy (SDP) proposed by Beimel, Nissim, and Omri (CRYPTO 2008) and Mironov, Pandey, Reingold, and Vadhan (CRYPTO 2009). In contrast to SDP, CDP only requires privacy guarantees to hold against computationally-bounded adversaries rather than computationally-unbounded statistical adversaries. Despite the question being raised explicitly in several works (e.g., Bun, Chen, and Vadhan, TCC 2016), it has remained tantalizingly open whether there is \emph{any} task achievable with the CDP notion but not the SDP notion. Even a candidate such task is unknown. Indeed, it is even unclear what the truth could be!

In this work, we give the first construction of a task achievable with the CDP notion but not the SDP notion, under the following strong but plausible cryptographic assumptions:
\begin{itemize}[leftmargin=*]
\item Non-Interactive Witness Indistinguishable Proofs,
\item Laconic Collision-Resistant Keyless Hash Functions,
\item Differing-Inputs Obfuscation for Public-Coin Samplers.
\end{itemize}
In particular, we construct a task for which there exists an $\varepsilon$-CDP mechanism with $\varepsilon = O(1)$ achieving $1-o(1)$ utility, but any $(\varepsilon, \delta)$-SDP mechanism, including computationally-unbounded ones, that achieves a constant utility must use either a super-constant $\varepsilon$ or an inverse-polynomially large $\delta$. 

To prove this, we introduce a new approach for showing that a mechanism satisfies CDP:
first we show that a mechanism is ``private'' against a certain class of decision tree adversaries, and then we use cryptographic constructions to ``lift'' this into privacy against computationally bounded adversaries.
We believe this approach could be useful to devise further tasks separating CDP from SDP.
\end{abstract}


\section{Introduction}
\label{sec:intro}

The framework of differential privacy (DP) \cite{DworkMNS06, DworkKMMN06} gives formal privacy guarantees on the outputs of randomized algorithms. It has been the subject of a significant body of research, leading to numerous practical deployments including the US census \cite{abowd2018us}, and industrial applications \cite{erlingsson2014rappor,CNET2014Google, greenberg2016apple,dp2017learning, ding2017collecting, LinkedINDP1, LinkedInDP2}.

The definition of DP requires privacy against
computationally unbounded, i.e., statistical, adversaries.
A natural modification is to instead only require privacy against computationally bounded adversaries. In cryptography, considering computationally bounded adversaries instead of statistical ones enables a vast array of applications, like public-key cryptography. Could the same be true for DP?  Despite Beimel, Nissim, and Omri~\cite{BeimelNO08} defining computational differential privacy (CDP) in 2008 (definitions that were further extended by Mironov, Pandey, Reingold, and Vadhan~\cite{MironovPRV09}), the central question of separating it from statistical differential privacy (SDP)\footnote{See \Cref{sec:prelim} for the formal definitions of CDP and SDP.  A good survey of the area can be found in \cite[Section 10]{Vadhan17}.}, in the standard client-server model, remains open:
\begin{question}
\cite[Open Problem 10.6]{Vadhan17}
\label{question:main}
\textsl{
Is there a computational task solvable by a single curator with computational differential privacy but is impossible to achieve with information-theoretic differential privacy?}%
\end{question}

There have been several positive and negative results towards resolving this question. In the positive direction, it is known that in the multi-party setting,  CDP is stronger than SDP~\cite{McGregorMPRTV10, MironovPRV09}. Roughly speaking, this is because secure multi-party computation enables many data curators to simulate acting as a single central curator, without compromising privacy. Still, the multi-party setting seems very different than the single-curator (aka central) setting. Indeed, \cite{McGregorMPRTV10} remark\footnote{This remark is also quoted by Groce, Katz, and Yerukhimovich~\cite{GroceKY11}.} that their ``\textsl{strong separation between (information-theoretic) differential privacy and computational differential privacy  ... stands in sharp contrast with the client-server setting where ... there are not even candidates for
a separation}.''

In the central setting, Bun, Chen, and Vadhan~\cite{BunCV16} show there is a task for which there is a CDP mechanism, but any SDP mechanism for this task must be inefficient (modulo certain cryptographic assumptions). We stress that the task they consider \emph{does} have an inefficient SDP mechanism (with parameters that match their CDP mechanism), so it does not resolve \cref{question:main}. While this may seem like a minor technical point, we emphasize that it is of crucial importance. Perhaps the main practical motivation behind studying CDP is the hope that there are CDP mechanisms for natural tasks with parameters that \emph{beat} the lower bounds against SDP mechanisms. But if, as in the case of the result in \cite{BunCV16}, there exists (even an inefficient) SDP mechanism matching the parameters of the CDP mechanism, then there is no hope of the CDP mechanism's parameters beating SDP lower bounds.

In the negative direction, Mironov, Pandey, Reingold, and Vadhan~\cite{MironovPRV09}  (building on Green and Tao~\cite{GreenT08}, Tao and Ziegler~\cite{TaoZ08}, and Reingold, Trevisan, Tulsiani, and Vadhan~\cite{ReingoldTTV08}) show a ``dense model theorem'' for pairs of random variables with ``pseudodensity'' with each other.  
Mironov et al.~\cite{MironovPRV09} note that (roughly speaking) extending this dense model theorem to handle multiple pairs of random variables would prove that any CDP mechanism could be converted into an SDP mechanism; such an extension is still open~\cite[Open Problem 10.8]{Vadhan17}.

Groce, Katz, and Yerukhimovich~\cite{GroceKY11} show that CDP mechanisms for certain tasks where the output is low-dimensional imply SDP mechanisms. Many natural statistical tasks fall into this category, and consequently, such tasks cannot separate CDP from SDP.  (This result was further strengthened by~\cite{BunCV16}.)  Furthermore,~\cite{GroceKY11} show that CDP mechanisms constructed in a black-box way from a variety of cryptographic objects, such as one-way functions, random oracles, trapdoor permutations, and cryptographic hash functions, cannot separate CDP from SDP.\\

\noindent In summary, there are at least two \emph{barriers} to separate CDP from SDP:
\begin{enumerate}[nosep,leftmargin=*]
    \item \textbf{High-dimensionality:} One needs to consider (perhaps non-natural) tasks with high dimensional outputs;
    \item \textbf{Exotic cryptography:} One needs to use cryptography somewhat specially (perhaps either an exotic primitive or in a non-black-box manner).
\end{enumerate}

\medskip
\noindent In light of these both positive and negative results as well as the lack of a candidate separation, it is not even clear what the truth could be: is there \emph{any} task for which there is a CDP mechanism but no SDP mechanism?\\[-3mm] 

\myparagraph{Our Contributions.}
We show, under plausible cryptographic hypotheses, that there are indeed tasks for which there exist CDP mechanisms but no SDP mechanisms.
 This not only positively answers \Cref{question:main} but also negatively answers the dense model extension question~\cite[Open Problem 10.8]{Vadhan17}.
 We state this result now informally and formalize it later in \Cref{sec:proof-overview}. We also delay discussing our precise cryptographic assumptions to \Cref{subsec: plausibility of crypto}, where we discuss their plausibility in detail.

\begin{theorem}\label{thm:cdp-vs-sdp-main}[Informal version of \Cref{thm:main-formal}]
Under cryptographic assumptions, there exists a task for which there is a $\CDP$ mechanism but no $\SDP$ mechanism.
\end{theorem}

Let us take a step back to discuss the implications of \Cref{thm:cdp-vs-sdp-main}. Although (as we will see in a moment) our task is specifically constructed for the purpose of separating CDP and SDP, the fact that we can separate them at all opens up a possibility that such a separation even holds for some ``natural'' tasks. Indeed, some of the current lower bound techniques for SDP---such as the ubiquitous ``packing lower bounds''\footnote{Specifically, when the packing lower bound requires the use of super-polynomially many datasets, the corresponding adversary does not necessarily run in polynomial time.} (see~\cite{HardtT10})---do not necessarily rule out CDP mechanisms. It seems prudent to carefully reexamine the current lower bound techniques to see whether they also apply to CDP. The ultimate hope for this program would be to employ CDP to overcome the known SDP lower bounds for some more ``natural'' tasks. (Of course, such tasks would also give a more ``natural'' separation of CDP and SDP.)

In fact, the technical approach we use in our construction already suggests a general approach for constructing non-trivial CDP mechanisms that could apply to more tasks. We discuss this in more detail in \Cref{sec:proof-overview}, but the idea is as follows. In order to show a task has a CDP mechanism, first show there is a mechanism for that task that is ``private'' against a certain class of decision tree adversaries. Then, second, use cryptographic assumptions to ``lift'' this into privacy against computational adversaries.\\[-3mm]

\myparagraph{Organization.}
The rest of the paper is organized as follows. \Cref{sec:proof-overview} provides a high-level overview of our techniques as well as a discussion of our cryptographic assumptions and their plausibility. \Cref{sec:prelim} contains the background material and \Cref{sec:prob} formally defines the problems.  We provide our CDP mechanism in \Cref{sec:CDP}, and prove lower bounds against SDP mechanisms in \Cref{sec:SDP}. These two components are put together to prove the main result in \Cref{sec:proof-main-thm}. Finally, we discuss the open problems and future directions in \Cref{sec:conclusion}.

\section{Overview of the Results}
\label{sec:proof-overview}

We will next discuss the high-level overview of our results and techniques. We will sometimes have to be informal here, but all details are formalized later. We first recall how a ``task'' is defined.\footnote{Refer to~\Cref{subsec:prelim-utility} for a more formal definition.} Following~\cite{GroceKY11,BunCV16}, a \emph{task} is defined by an efficiently computable \emph{utility} function $u$ that takes in an input dataset $D$ and a response $y$ such that $u(D, y) = 1$ if $y$ is considered ``useful'' for $D$ and $u(D, y) = 0$ otherwise. A mechanism $M$ is said to be \emph{$\alpha$-useful} for $u$ iff $\E[u(D, M(D))] \geq \alpha$ for all input datasets $D$; 
we will refer to $\alpha$ as the \emph{usefulness} of $M$.
We remark that many well-studied problems---such as linear queries with various error metrics---can be written in this form.

One of our main conceptual contributions is to define a class of tasks that seems to naturally circumvent the two earlier-mentioned barriers---tasks where one needs to output a \emph{circuit.}

\subsection{The Low Diameter Set Problem}

Before we detail why tasks that output a circuit might evade the two barriers, let us describe a concrete example. We call the following the {\em low diameter set} $\LDS_{\tau}$ problem (defined for some parameter $\tau \in \N$): 
\begin{itemize}[leftmargin=*]
    \item \textbf{Given:} dataset $D$ represented as $n$ bits (adjacent datasets differ on a single bit)\footnote{Refer to \Cref{sec:prob} for formal details.}
    \item \textbf{Output:} circuit $C$ mapping $n$ bits to $1$ bit
    \item \textbf{Utility:} $C$ is considered useful %
    if it outputs \begin{itemize}[nosep,label=$\triangleright$]
        \item $1$ on $D$, and
        \item $0$ on all points at distance greater than $\tau$ from $D$.%
    \end{itemize}
\end{itemize}
Informally, this problem asks to output a circuit $C$ such that $C^{-1}(1) \subseteq \bit^n$ contains $D$ and has diameter at most $\tau$.
While this utility function is not efficiently computable, we will address this in \Cref{para:utility}.
Looking ahead, we will ultimately separate CDP from SDP under cryptographic assumptions by considering a ``verifiable'' version of this problem where we only care about datasets in a cryptographically special set.\\[-3mm]

\noindent We now revisit the two barriers and discuss how the distance problem might circumvent them. 
\begin{enumerate}[nosep,leftmargin=*]
\item \textbf{High-dimensionality:} The output of this task is a circuit, which is high-dimensional.
\item \textbf{Exotic cryptography:} Because the output of the task is a circuit, it lends itself to a powerful class of cryptographic objects: \emph{circuit obfuscators} \cite{BarakGIRSVY12}. Roughly speaking, circuit obfuscators take as input a circuit $C$ and output a scrambled, obfuscated circuit $C'$ that computes the same function as $C$ but which, ideally, has the property that ``anything you could do with access to the circuit $C'$, you could do with only black-box access to the function the circuit computes.'' Importantly, obfuscation is \emph{not} in the list of primitives ruled out by the barrier in~\cite{GroceKY11}.
\end{enumerate}

\subsection{SDP Lower Bound}

Our starting point for separating CDP from SDP is the low diameter set problem described above. Indeed, we show that there is no SDP mechanism for this problem for any $\tau$ that is essentially sub-linear in $n$.
\begin{lemma}\label{lem:blatant-lb}
For all $0 < \tau \leq n^{0.9}$ and constant $\eps, \alpha > 0$ and $\delta = 1/n^{c}$ (for some $c > 1$), there is no $(\eps, \delta)$-SDP mechanism for $\LDS_{\tau}$ that is $\alpha$-useful.
\end{lemma}
\noindent In fact, this lower bound is straightforward (\Cref{lem:low-diameter-to-nearbypoint}) from the well-known \emph{blatant non-privacy} notion (see, e.g.,~\cite{De12}): no DP algorithm can output a dataset that is (with large probability) close to the input dataset.
Crucially, our lower bounds are non-constructive, and do not yield an efficient adversary (which would imply a similar lower bound against CDP mechanisms).  Thus, to separate CDP from SDP it suffices to come up with a CDP mechanism for, say, $\LDS_{n^{0.9}}$. 

\subsection{A CDP Mechanism}

By \Cref{lem:blatant-lb}, a positive answer to the following question would demonstrate a separation between CDP and SDP.

\begin{question}\label{ques:distance-main}
For constant $\eps = O(1)$, does there exist an $\eps$-$\CDP$ mechanism for $\LDS_{n^{0.9}}$ with constant usefulness?
\end{question}

\noindent A key step in our approach is to reduce the above question to whether there is a mechanism for $\LDS_{n^{0.9}}$ that is differentially private against query (a.k.a. decision-tree) adversaries. In order to construct such a CDP mechanism $M$, our main idea is to use obfuscation. In particular, we will consider mechanisms where the returned circuit is obfuscated. Recall that in order to prove a mechanism $M$ that outputs a circuit $C$ is CDP, one needs to argue that no efficient adversary that gets $C$ as input can break the privacy guarantee. By considering mechanisms that return obfuscated circuits, we can drastically simplify the type of adversaries we need to prove privacy against. Instead of proving privacy against adversaries that see the circuit $C$ (i.e., white-box setting), sufficiently strong obfuscation means we only need to prove privacy against decision tree adversaries that can query the function computed by the circuit (i.e., black-box setting). In other words, if we have a mechanism that satisfies DP against black-box adversaries (decision trees) with a polynomial number of queries, we can then hope to use sufficiently strong obfuscation to ``lift'' this into a mechanism that is secure against (white-box) computational adversaries with polynomial running time.

Of course, one needs to be careful about whether such ``sufficiently strong obfuscation'' is even possible, but, putting that aside for the moment, the question of whether there is a CDP mechanism for $\LDS_{n^{0.9}}$ (\Cref{ques:distance-main} above) appears to reduce to whether there is a mechanism for $\LDS_{n^{0.9}}$ that is DP against query (a.k.a. decision-tree) adversaries.

While we do not resolve \Cref{ques:distance-main}, we (roughly speaking) show that there is a mechanism that is DP against \emph{non-adaptive} decision tree adversaries, whose queries are fixed a priori. It turns out a relatively simple mechanism based on randomized response~\cite{warner1965randomized} works for these less powerful adversaries.

\subsubsection{From Non-Adaptive Lower Bound to Computational Lower Bound}
This switch from the usual adaptive query adversaries to non-adaptive query adversaries comes at a price however. It is not clear how to use obfuscation to lift a mechanism that is private against non-adaptive queries into one that is private against computational adversaries. Indeed, a polynomial-time algorithm with even black-box access to a function seems to be an inherently adaptive adversary! 

Surprisingly, we get around this by using another cryptographic object introduced by Bitansky, Kalai, and Paneth~\cite{BitanskyKP18}: collision-resistant keyless hash functions. Informally speaking, a hash function 
being collision-resistant and keyless means that ``any efficient adversary can only generate a number of hash collisions that is at most polynomially larger than the advice the adversary gets.''

We then modify the $\LDS_{\tau}$ problem to only consider datasets that belong to a specific set $\cR \subseteq \bit^n$; in particular, we will specify it the as set of all strings that hash to, say the all zeroes string. Formally, $\LDS_{\tau, \cR}$ is the following problem (defined for parameters $\tau \in \N$ and $\cR \subseteq \bit^n$).
\begin{itemize}[leftmargin=*]
    \item \textbf{Given:} dataset $D$ that consists of $n$ bits
    \item \textbf{Output:} circuit $C$ mapping $n$ bits to $1$ bit 
    \item \textbf{Utility:} $C$ is considered useful if $D \notin \cR$ or both of the following hold:
    \begin{itemize}[nosep,label=$\triangleright$]
        \item it outputs 1 on $D$
        \item it outputs 0 on all points in $\cR$ at distance greater than $\tau$ from $D$
    \end{itemize}
\end{itemize}
In other words, the utility function now completely ignores all points outside of $\cR$.\\[-3mm]

\noindent The high-level intuition behind this change is the following:
\begin{enumerate}[leftmargin=*]
    \item Our CDP mechanism can output a circuit $C$ such that the only inputs where $C(x)$ reveals information are those $x$ in the set $\cR$. %
    \item Any polynomial-time adversary $\cA$ can only generate fixed polynomial number of elements of $\cR$ by the collision-resistance property of the hash function.
    \item Combining the above makes the inputs $\cA$ can query $C$ on, effectively ``non-adaptive''.
\end{enumerate}

\noindent Finally, in order to ``lift'' the query separation into the computational realm we use another cryptographic tool: differing-inputs obfuscation ($\diO$)~\cite{BarakGIRSVY01,BarakGIRSVY12,AnanthBGSZ13}. Roughly speaking, $\diO$ is an obfuscator with the following guarantee: if any efficient adversary can distinguish the obfuscation of two circuits $C_1$ and $C_2$, then an efficient adversary can find an input $x$ on which $C_1(x) \neq C_2(x)$. In particular, the specific assumption we use is even weaker than public-coin $\diO$~\cite{IshaiPS15}, which is already considered to more plausible than general $\diO$.\footnote{See \Cref{as:diO} for formal statement of the assumption and \Cref{apx:diO-comparison} for comparison with other $\diO$ assumptions in literature.}

In summary, $\diO$ allows us to reduce computational adversaries to adaptive query adversaries and collision-resistant keyless hash functions allows us to reduce adaptive query adversaries to non-adaptive query adversaries. Interestingly, to the best of our knowledge, this is the first time collision-resistant keyless hash functions are being used together with any obfuscation assumption.

\subsubsection{Making the Utility Function Efficiently Computable}\label{para:utility}
We need to address one final issue: utility functions that we have considered so far are not necessarily efficiently computable. Specifically, a trivial way to implement the utility function would be to enumerate all points at distance at least $\tau$, feed it into the circuit, and check that the output is as expected; this would take $2^{n^{\Omega(1)}}$ time.

To overcome the above problem, we restrict circuits to only those that are relatively simple, so that there is a small ``witness'' $w$ that certifies that the circuit outputs zero at all points that are $\tau$-far from $D$. A naive idea is then to let the CDP mechanism output the circuit $C$ together with such a witness $w$. The utility function can then just efficiently check that $w$ is a valid witness (and that $C(D) = 0$ or $x \in \cR$). This makes the utility function efficient but unfortunately compromises privacy because the witness $w$ itself can leak additional information. To avoid this, we instead use non-interactive witness indistinguishable (NIWI) proofs (e.g., \cite{BarakOV07}). Roughly speaking, this allows us to produce a proof $\pi$ from $w$ (and $C$ and $\diO$), which does not leak any information about $w$ (against computationally bounded adversaries), but at the same time still allows us to verify that the underlying witness $w$ is valid. The former is sufficient for CDP, while the latter ensures that the utility function can be computed efficiently.

This completes the high-level overview of the constructed task and our $\CDP$ mechanism.  The cryptographic primitives needed for our mechanism are formalized in \Cref{as:crklhf,as:diO,as:niwi}.

\subsection{Final Steps}

Finally, since our problem is now not exactly the original $\LDS_{\tau}$ problem anymore, as the utility guarantees are only now meaningful for datasets in $\cR$, we cannot use the lower bound in \Cref{lem:blatant-lb} for $\LDS_{\tau}$ directly. Fortunately, we can still adapt its proof---a ``packing-style'' lower bound on each coordinate---to one which applies a packing-style argument on each \emph{block of coordinates} instead. With this, we can prove the lower bound for $\LDS_{\tau, \cR}$ as long as the set $\cR$ has sufficiently large density ($\approx 1/n^{-o(\log n)}$).

Putting all the ingredients together, we arrive at the following\footnote{We remark that $(\beps_{\SDP}, \bdelta_{\SDP})$-$\SDP$ mechanism here refers to an ensemble of mechanisms $\seq{M_n}$ that are $({\beps_{\SDP}}_n, {\bdelta_{\SDP}}_n)$-$\SDP$. (See \Cref{def:sdp}.)}:
\begin{theorem}[Main Result] \label{thm:main-formal}
Under \cref{as:crklhf,as:diO,as:niwi}, for any constant $\beps_{\CDP} > 0$, there exists an ensemble $\bu = \set{u_n}_{n \in \N}$ of polynomial time computable utility functions such that
\begin{itemize}[nosep,leftmargin=*]
\item There is an $\beps_{\CDP}$-$\CDP$ mechanism that is $(1 - o_n(1))$-useful for $\bu$.
\item For any constants $\beps_{\SDP}, \balpha > 0$ and $\bdelta_{\SDP} = 1/n^{27}$, no $(\beps_{\SDP}, \bdelta_{\SDP})$-$\SDP$ mechanism is $\balpha$-useful for $\bu$.
\end{itemize}
\end{theorem}
\noindent The task underlying the separation is an instantiation of the ``verifiable low diameter set problem'' $\VLDS_{\tau,\cR,V}$ defined in \Cref{def:vlds}.

\subsection{On the Plausiblility of the Cryptographic Assumptions}\label{subsec: plausibility of crypto}

We now discuss the plausiblility of the three cryptographic assumptions we use for our result:
\begin{enumerate}[label=(\roman*),leftmargin=7mm]
\item NIWI: Non-interactive Witness Indistinguishable Proofs (formally, \cref{as:niwi})
\item CRKHF: Laconic Collision-Resistant Keyless Hash Functions (formally, \cref{as:crklhf})
\item $\diOpcS$: Differing-Inputs Obfuscation for Public-Coin Samplers (formally, \cref{as:diO})
\end{enumerate}

\myparagraph{Regarding (i), NIWI.} Bitansky and Paneth~\cite{BitanskyP15} show that NIWIs exist assuming one-way permutations and indistinguishability obfuscation (iO) exists. Recently, Jain, Lin, and Sahai \cite{JainLS21} show that the existence of $\iO$ follows from well-founded assumptions; consequently, NIWIs exist based on widely-believed assumptions.
(We note that other previous works have also constructed NIWIs based on other more specific assumptions~\cite{BarakOV07,GrothOS12}.)

\myparagraph{Regarding (ii), CRKHF.} Bitansky, Kalai, and Paneth~\cite{BitanskyKP18} defined CRKHFs to model the properties of existing hash functions like SHA-2 used in practice. They suggest several candidates for CRKHFs, such as hash functions based on AES and Goldreich's one-way functions. They also note that CRKHFs exist in the Random Oracle model, as a random function is a CRKHF. Still, it is an open question to base the security of a CRKHF on a standard cryptographic assumption. Part of the difficulty of doing this, as~\cite{BitanskyKP18} describe, is that most cryptographic assumptions involve some sort of structure that is useful for constructing cryptographic objects. In contrast, the goal of a CRKHF is to have no structure at all. In summary, given the various CRKHF candidates, the existence in the Random Oracle model, and the fact that CRKHFs exist ``in practice,'' this assumption is quite plausible. For our specific construction, we need a different hash length (equivalently, different compression rate) than that used in~\cite{BitanskyKP18}; please refer to the discussion preceding~\Cref{as:crklhf} for the parameters and justification.

Finally, we remark that, even though the existence of CRKHFs is not known to reduce to any ``well-founded'' assumption, even \emph{refuting} their existence would answer a longstanding question in cryptography: giving non-contrived separations between the Random Oracle model~\cite{BellareR93} and the standard model. In the words of Bitansky, Kalai, and Paneth~\cite{BitanskyKP18},
\begingroup
\addtolength\leftmargini{-5mm}
\begin{quote}
``\textsl{Any attack on the multi-collision resistance of a [keyless] cryptographic hash function would constitute a strong and natural separation between
the hash and random oracles. For several cryptographic hash functions used in practice, the only known separations from random oracles are highly contrived~\cite{CanettiGH04}.}''
\end{quote}
\endgroup

\myparagraph{Regarding (iii), $\diOpcS$.} One can think of $\diO$~\cite{BarakGIRSVY01, BarakGIRSVY12} as an ``extractable'' strengthening of $\iO$.  While $\iO$ has now become a widely-believed assumption \cite{JainLS21}, the existence of $\diO$ is controversial. Several papers (e.g.,~\cite{BoyleP15,GargGHW17,BellareSW16}) cast doubt on the existence of $\diO$, especially in the case where an arbitrary auxillary input is allowed; we stress that all the negative results for $\diO$ hold for contrived auxillary inputs and/or distributions.
On the positive side, \cite{BoyleCP14} show that $\diO$ reduces to $\iO$ in special cases, such as when the number of differing-inputs is bounded by a polynomial. More related to our result, \cite{IshaiPS15} gives a definition of {\em public-coin $\diO$} that avoids the difficulties presented by earlier negative results regarding auxiliary inputs, although \cite{BoyleP15} presented some evidence against this definition in special cases.
Our specific assumption of $\diOpcS$ is in fact weaker than the assumption of public-coin $\diO$. In the definition of public-coin $\diO$, as in \cite{IshaiPS15}, we start with any public-coin sampler ($\pcS$), for which it is hard to find an input on which two circuits differ, even given the knowledge of all the randomness that underlies the circuits. The security of the obfuscation is required to hold even against adversaries that know all the randomness that underlies the generation of the two circuits. However, in our definition, the security of the obfuscation is required to hold only against adversaries that observes a single obfuscated circuit, which makes the assumption weaker. See \Cref{apx:diO-comparison} for a more detailed discussion on comparison of this assumption with other $\diO$ assumptions in literature.
Finally, we only use the existence of $\diOpcS$ for a simple circuit family for our result, so even if general purpose $\diOpcS$ does not exist, we think it is plausible that $\diOpcS$ exists for the specific family of circuits we need for our result. (See \Cref{as:diO} for the exact $\pcS$ family for which we require a $\diO$.)

\paragraph{Final thoughts on our assumptions.}
In conclusion, we view each of our three assumptions as {\em plausible}. Moreover, each of assumptions has at least some evidence that is hard to refute: NIWIs exist based on a widely-believed assumption, refuting CRKHFs would require giving the first non-contrived separation between the standard and the Random Oracle model, and despite many attempts (e.g.,~\cite{BoyleP15,GargGHW17,BellareSW16}) to refute $\diO$, the question is still open, especially for the particular $\diOpcS$ version. Thus, refuting any of the assumptions would constitute a breakthrough in cryptography.

\section{Preliminaries}
\label{sec:prelim}

A function $g: \N \to \R_{\geq 0}$ is said to be \emph{negligible} if $g(n) = n^{-\omega(1)}$.
Let PPT be an abbreviation for {\bf p}robabilistic {\bf p}olynomial-time {\bf T}uring machine.

For $x \in \bit^n$ and $r \in \N$, we use $B_r(x)$ to denote the (Hamming) ball of radius $r$ around $x$, i.e., $\{z \in \bit^n \mid \|x - z\|_1 \leq r\}$. Furthermore, we use $\diam(S)$ for a set $S \subseteq \bit^n$ to denote the (Hamming) diameter of $S$, i.e., $\max_{x, x'\, \in S} \|x - x'\|_1$. 

\subsection{Dataset and Adjacency}

For a domain $\cX$, we view a dataset $D$ as a histogram over the domain $\cX$, i.e., $D \in \Z_{\geq 0}^{\cX}$ where $D_x$ denotes the number of times $x \in \cX$ appears in the dataset. The \emph{size} of the dataset is defined as $\|D\|_1 := \sum_{x \in \cX} D_x$. We write $\cX^m$ as a shorthand for the set of all datasets of size $m$, and $\cX^*$ for the set of all datasets over domain $\cX$. Two datasets are \emph{adjacent} iff $\|D - D'\|_1 = 1$, i.e., one of the datasets is a result of adding or removing a single row from the other dataset.

\subsection{Mechanism, Utility Function, and Usefulness}\label{subsec:prelim-utility}

A \emph{mechanism} $M$ is a randomized algorithm that takes in a dataset $D \in \cX^*$ and outputs an element from a set $\cY$. The utility of a mechanism is measured by a \emph{utility function} $u$, which is a polynomial-time deterministic algorithm that takes in a dataset $D \in \cX^*$ together with a response $y \in \cY$ and outputs 0 or 1 (whether the response is good for the dataset). We say that the mechanism $M$ is \emph{$\alpha$-useful} for utility $u$ iff $\Pr[u(D, M(D)) = 1] \geq \alpha$.

Below, we will often discuss an ensemble $\bM = \{M_n\}_{n \in \N}$ of mechanisms  where\footnote{It is always implicitly assumed that $\cX_n, \cY_n$ are of size $\poly(n)$.}
$M_n: \cX_n^* \to \cY_n$. We say that an ensemble of mechanisms is \emph{efficient} if $M_n$ on input $D \in \cX_n^m$ runs in time $\poly(n, m)$. For an ensemble $\bu = \seq{u_n}_{n \in \N}$ of utility functions and $\balpha = \seq{\alpha_n \in [0, 1]}_{n \in \N}$, we say that $\bM$ is $\balpha$-useful with respect to $\bu$ iff $M_n$ is $\alpha_n$-useful with respect to $u_n$ for all $n \in \N$.

For brevity, we will sometimes refer to ``ensemble of mechanisms'' simply as ``mechanism'' and ``ensemble of utility functions'' simply as ``utility function'' when there is no ambiguity.

\subsection{Notions of Differential Privacy}

We now define the notions of DP that will be used throughout the paper.\\[-3mm]

\myparagraph{(Statistical) Differential Privacy.} The standard (statistical) notion of DP can be defined in terms of the following notion of indistinguishability.

\begin{definition}[Statistical Indistinguishability]
Distributions $P$, $Q$ are said to be {\em $(\eps,\delta)$-indistinguishable}, denoted $P \approx_{\eps, \delta} Q$, if for all events (measurable sets) $\cE$, it holds for $(\cD_0, \cD_1) = (P, Q)$ and $(Q, P)$ that
\begin{align*}
\Pr_{X \sim \cD_0} [X \in \cE] &~\le~ e^{\eps} \cdot \Pr_{X \sim \cD_1} [X \in \cE] + \delta\,.%
\end{align*}
For simplicity, we use $\approx_{\eps}$ to denote $\approx_{\eps, 0}$.
\end{definition}

\begin{definition}[Statistical Differential Privacy (SDP)~\cite{DworkMNS06,DworkKMMN06}]
\label{def:sdp}
For $\eps, \delta > 0$, a mechanism $M$ is said to be \emph{$(\eps, \delta)$-$\SDP$}  if and only if for every pair $D, D'$ of adjacent datasets, we have that $M(D) \approx_{\eps, \delta} M(D')$.
We say that an ensemble $\bM = \{M_n\}_{n \in \N}$ is $(\beps, \bdelta)$-$\SDP$ for sequences $\beps = \seq{\eps_n}_{n \in \N}$ and $\bdelta = \seq{\delta_n}_{n \in \N}$ if $M_n$ is $(\eps_n, \delta_n)$-$\SDP$ for all $n \in \N$.
\end{definition}

\myparagraph{Computational Differential Privacy.} The notion of computational DP relaxes the notion of indistinguishability to a computational version, where the privacy holds only with respect to computationally bounded adversaries.

\begin{definition}[Computational Indistinguishability]
Two ensembles of distributions $\bP = \seq{P_n}_{n \in \N}$ and $\bQ = \seq{Q_n}_{n \in \N}$, where $P_n$ and $Q_n$ are supported over $\bit^{p(n)}$ for some polynomial $p(\cdot)$, are said to be {\em $\beps$-computationally-indistinguishable} for a sequence $\beps = \seq{\eps_n}_{n \in \N}$, denoted $\bP \approx^c_{\beps} \bQ$, if there exists a negligible function $\negl(\cdot)$ such that for any PPT adversary $\cA$, it holds for $(\cD_0, \cD_1) = (P_n, Q_n)$ and $(Q_n, P_n)$ that
\begin{align*}
\Pr_{X \sim \cD_0} [\cA(X) = 1] &~\le~ e^{\eps_n} \Pr_{X \sim \cD_1} [\cA(X) = 1] + \negl(n)
\end{align*}
In the special case of $\beps = 0$, we suppress the subscript and simply write $\bP \approx^c \bQ$.
\end{definition}

\noindent Throughout, when we refer to a sequence $\{(D_n, D'_n)\}_{n \in \N}$ of adjacent datasets, it is always assumed  that $D_n \in \cX_n^{m_n}, D'_n \in \cX_n^{m'_n}$ are of sizes $m_n, m'_n = \poly(n)$.

\begin{definition}[Computational Differential Privacy ($\CDP$){~\cite{MironovPRV09}}]
An ensemble $\bM = \seq{M_n}_{n \in \N}$ of mechanisms is said to be \emph{$\beps$-$\CDP$} for a sequence $\beps = \seq{\eps_n}_{n \in \N}$, if for any sequence $\seq{(D_n, D'_n)}_{n \in \N}$ of adjacent datasets, it holds that $\seq{M_n(D_n)}_{n \in \N} \approx^c_{\eps_n} \seq{M_n(D_n')}_{n \in \N}$.
\end{definition}

This definition is often referred to as \emph{indistinguishability-based CDP} ($\mathsf{IND}$-$\CDP$) in previous works~\cite{MironovPRV09,GroceKY11,BunCV16}. Since we only use this notion for our main result, we refer to it simply as CDP.  The other definition of CDP used in previous works is simulation-based:

\begin{definition}[$\mathsf{SIM}$-$\CDP${~\cite{MironovPRV09}}]
An ensemble $\bM = (M_n)_{n \in \N}$ of mechanisms is said to be \emph{$\beps$-$\mathsf{SIM}$-$\CDP$} if there exists an $(\eps_n, 0)$-$\SDP$ ensemble $\{ M'_n \}_{n \in \N}$ of mechanisms such that for any sequence $\{D_n \in \cX_n^*\}_{n \in \N}$ of datasets, with size of $D_n$ being at most $\poly(n)$, it holds that $M_n(D_n) \approx^c M'_n(D_n)$.
\end{definition}

It should be noted that $\mathsf{SIM}$-$\CDP$ \emph{cannot} be used for the separation we are looking for. Specifically, if $\{M_n\}_{n \in \N}$ is $\beps$-$\mathsf{SIM}$-$\CDP$, we may use $\{M'_n\}_{n \in \N}$ as our $(\beps, {\boldsymbol{0}})$-$\SDP$ mechanism. Since the utility function runs in polynomial time, it follows immediately that, if $\{M_n\}_{n \in \N}$ is $\alpha$-useful, then $\{M'_n\}_{n \in \N}$ is also $(\alpha - o(1))$-useful.
Due to this, we will not consider $\mathsf{SIM}$-$\CDP$ again in this paper.\\[-3mm]

\noindent Another point to note is that unlike prior work (e.g. \cite{BunCV16}) we use both $(\beps, \bdelta)$ parameters for $\SDP$, but only $\beps$ parameter for $\CDP$, since $\bdelta$ is always assumed to be negligible for $\CDP$. Our lower bounds for $\SDP$ in fact work for $\bdelta$ that is not negligible, which only makes the result stronger.

\myparagraph{Calculus of $\approx$ and $\approx^c$.} The following properties are well-known.

\begin{fact}\label{fact:dp-calculus}
The notions of $(\eps,\delta)$-indistinguishability and $\eps$-computational-indistinguishability satisfy:
\begin{itemize}[nosep,leftmargin=*]
\item {\bf Basic Composition:} If $P_0 \approx_{\eps, \delta} P_1$ and $P_1 \approx_{\eps', \delta'} P_2$, then $P_0 \approx_{\eps+\eps', \delta+\delta'} P_2$. Similarly, if $\bP_0 \approx^c_{\beps} \bP_1$ and $\bP_1 \approx^c_{\beps'} \bP_2$, then $\bP_0 \approx^c_{\eps+\eps'} \bP_2$.
\item {\bf Post-processing:} If $P \approx_{\eps, \delta} Q$, then for all (randomized) functions $f$, it holds that $f(P) \approx_{\eps, \delta} f(Q)$. Similarly, if $\bP \approx^c_{\beps} \bQ$, then for all PPT algorithms $\cA$, it holds that $\cA(\bP) \approx^c_{\beps} \cA(\bQ)$.
\end{itemize}
\end{fact}

\section{Low Diameter Set Problem and Nearby Point Problem}
\label{sec:prob}

In this section, we introduce the problems that we will use in our separation.  Before that, we will describe a simplifying assumption that we can make about the inputs.

\subsection{Simplification of Input Representation}

Recall that so far a dataset may contain multiple copies of an element. Below, however, it will be more convenient to only discuss the case where each element appears only once, i.e., $D \in \bit^{\cX}$.

This is sufficient since if we have a utility function $u: \bit^{\cX} \times \cY \to \bit$ defined only on $D \in \bit^{\cX}$, we can easily define the utility function $\ou: \N^{\cX} \times \cY \to \bit$ by
\begin{align*}
\ou(\oD, r) =
\begin{cases}
u(\oD, r) & \text{ if } \oD \in \bit^{\cX}, \\
1 & \text{ otherwise.}
\end{cases}
\end{align*}
In other words, the utility function considers any response good for datasets with repetition. Clearly, if $u$ is efficiently computable, then so is $\ou$.
Furthermore, suppose that we have an $\eps$-$\CDP$ mechanism $\bM = \seq{M_n}_{n \in \N}$ for $\bu = \seq{u_n}_{n \in \N}$. For every dataset $\oD$, let $D$ be defined by $D_i = \min\set{\oD_i, 1}$. Then, we may define $\obM = \seq{\oM_n}_{n \in \N}$ by $\obM(\oD) = M(D)$. It is easy to see that $\obM$ remains $\beps$-CDP. Furthermore, if $\bM$ is $\balpha$-useful for $\bu$, then $\obM$ remains $\balpha$-useful for $\obu$.

Finally, note that a lower bound for DP algorithms restricted to non-repeated datasets trivially implies a lower bound against all datasets.

Due to this, we will henceforth focus our attention only on the datasets $D \in \bit^{\cX}$. Furthermore, throughout the remainder of this paper, we will always pick $\cX_n = [n]$. This further simplifies the input representation to be just a bit vector $x \in \bit^n$. We will define an input of our problem in this way. Furthermore, we will henceforth use $x$ instead of $D$ to denote the input dataset.

\subsection{Nearby Point Problem}

We will start by defining our first problem, which asks to output a point that is close to the input point if the latter belongs to some set $\cR$. As we noted in the introduction, when $\cR$ is the set of all points (i.e., $\cR_n = \bit^n$), this is exactly the same as the problem considered in blatant non-privacy~\cite{DinurN03,DworkMT07}.
As we will see later, the presence of the set $\cR$ is due to our use of hashing, which is required in our proof for the $\CDP$ mechanism.

\begin{definition}[$\tau$-Nearby $\cR$-Point Problem]\label{def:nearby-point}
The \emph{nearby point problem} parameterized by sequences $\seq{\tau_n \in \N}_{n \in \N}$ and $\seq{\cR_n \subseteq \bit^n}_{n \in \N}$ is denoted by $\NBP_{\tau, \cR}$.  For input $x \in \bit^n$ and output $y \in \cY_n = \bit^n$, the utility is defined as:
\[ 
u^{\NBP}_{\tau_n, \cR_n}(x, y) := \indicator\set{\|x - y\|_1 \leq \tau_n \text{ or } x \notin \cR_n}
\]
\end{definition}

\noindent For brevity, we will assume throughout that $\cR_n$ is efficiently recognizable and henceforth we do not state this explicitly. Note that this assumption implies that the utility function defined above is efficiently computable.
The nearby point problem will be primarily used for proving the lower bounds against $\SDP$.

\subsection{Verifiable Low Diameter Set Problem}

Next, we define circuit-based tasks for which we will give $\CDP$ mechanisms.
To do so, we need to first define a ``$\tau$-diameter verifier''.

\begin{definition}[$\tau$-Diameter Verifier]
For a sequence $\tau = \seq{\tau_n}_{n \in \N}$  of integers, we say that an efficiently computable (deterministic) verifier $V = \seq{V_n}_{n \in \N}$ is a \emph{$\tau$-diameter verifier} for circuits of size $s(n)$ if it takes as input a circuit $C: \bit^n \to \bit$ of (polynomial) size $s(n)$ and a proof $\pi$ of size $\poly(n)$, and outputs $V_n(C, \pi) = 1$ only if $\diam(C^{-1}(1)) \leq \tau_n$.
\end{definition}

\noindent We can now define the (verifiable) low diameter set problem as follows:
\begin{definition}[Verifiable $\tau$-Diameter $\cR$-Set Problem]\label{def:vlds}
The \emph{verifiable low diameter set problem} parameterized by sequences $\tau = \seq{\tau_n}_{n \in \N}$, $\cR = \seq{\cR_n \subseteq \bit^n}_{n \in \N}$, and $\tau$-diameter verifier $V = \seq{V_n}_{n \in \N}$ is denoted by $\VLDS_{\tau, \cR, V}$. The input, output, and utility are defined as follows:
\begin{itemize}[leftmargin=*]
\item \textbf{Input:} $x \in \bit^n$.
\item \textbf{Output:} circuit $C$ and a proof $\pi$, both of size $\poly(n)$.
\item \textbf{Utility:} $u^{\VLDS}_{\tau_n, \cR_n, V_n}(x, (C, \pi)) := $ $\indicator\set{V_n(C, \pi) = 1}$ and $\indicator\set{C(x) = 1 \mbox{ or } x \notin \cR_n}$.
\end{itemize}
\end{definition}

\noindent For convenience, we also define the following utility function 
\begin{align*}
u^{\eval}_{\cR}(x, C) := \indicator\set{C(x) = 1 \mbox{ or } x \notin \cR}.
\end{align*}
Note that this does not correspond to a hard task, because a circuit that always outputs one is 1-useful. Nonetheless, it will be convenient to state usefulness of some intermediate algorithms via this utility function.

\subsection{From Low Diameter Set Problem to Nearby Point Problem}

Below we provide a simple observation that reduces  the task of proving an $\SDP$ lower bound for the verifiable low diameter set problem to that of the nearby point problem. (Note here that the $\SDP$ mechanisms considered below can be computationally inefficient.)

\begin{lemma} \label{lem:low-diameter-to-nearbypoint}
If there is an $(\beps, \bdelta)$-$\SDP$ $\balpha$-useful mechanism for the $\VLDS_{\tau, \cR, V}$ problem, then there is an $(\beps, \bdelta)$-$\SDP$ $\balpha$-useful mechanism for the $\NBP_{\tau, \cR}$ problem.
\end{lemma}

\begin{proof}
Let $M$ be an $(\beps, \bdelta)$-SDP $\balpha$-useful mechanism for the $\VLDS_{\tau, \cR, V}$ problem.
We will construct an $(\beps, \bdelta)$-SDP $\balpha$-useful mechanism $M'$ for the $\NBP_{\tau, \cR}$ problem.

The mechanism $M'_n$ on input dataset $x \in \bit^n$ works as follows. First, let $(C, \pi) \gets M_n(x)$. If $V_n(C, \pi) = 1$, then output the lexicographically first element of $C^{-1}(1)$ (else, output $0^n$). This completes our description of $M'$.

Since $M$ is $(\beps,\bdelta)$-SDP, we have that $M'$ is also $(\beps, \bdelta)$-SDP by post-processing. It remains to show that $M'$ is $\balpha$-useful. Fix some input $x \in \bit^n$. If $x \notin \cR_n$, then any output satisfies utility. Thus, it suffices to consider the case where $x \in \cR_n$. With probability $\alpha_n$, we have that $V_n(C, \pi) = 1$ (which implies that $C^{-1}(1)$ has diameter at most $\tau_n$), and $x \in C^{-1}(1)$. Consequently, the distance between $x$ and the lexicographically first element of $C^{-1}(1)$ is at most $\tau_n$. So with probability at least $\alpha_n$, the output of $M'$ is useful for $x$, as desired.
\end{proof}

\section{\texorpdfstring{$\CDP$}{CDP} Mechanism for Verifiable Low Diameter Set Problem}
\label{sec:CDP}

In this section we build a CDP mechanism for the verifiable low diameter set problem.  We establish the following result:
\begin{theorem} \label{thm:cdp-main}
Suppose that \Cref{as:crklhf,as:diO,as:niwi} hold.
Then, for all constant $\beps_{\CDP} > 0$ and $\tau = \seq{\tau_n = n^{0.9}}_{n \in \N}$, there exists a $\tau$-diameter verifier $V$ and a sequence $\cR = \seq{\cR_n}_{n \in \N}$ of sets of sizes $|\cR_n| \ge 2^n / n^{o(\log n)}$, such that there exists an $\beps_{\CDP}$-$\CDP$ mechanism that is $(1 - o_n(1))$-useful for $u^{\VLDS}_{\tau, \cR, V}$.
\end{theorem}

As discussed in the overview, we first build a mechanism that is $\CDP$ but without verifiability using collision-resistant keyless hash functions and differing-inputs obfuscators (\Cref{sec:cdp-unverifiable}). We then turn it into a verifiable one using non-interactive witness indistinguishable proofs (\Cref{sec:cdp-verifiable}).

\subsection{\texorpdfstring{$\CDP$}{CDP} Mechanism without Verifiability}
\label{sec:cdp-unverifiable}

In this section, we construct our first $\CDP$ mechanism (\Cref{alg:dio}). We depart from the  overview in \Cref{sec:proof-overview}  slightly and do not prove a non-adaptive query lower bound explicitly. Instead, we directly show in \Cref{subsec:di-circuit-family} how to sample the appropriate differing-inputs circuit family. This can be then easily turned into our $\CDP$ mechanism via $\diO$ in \Cref{subsec:cdp-unverifiable}.

\subsubsection{Additional Preliminaries: Cryptographic Primitives}

Throughout this section, we will repeatedly use the so-called \emph{randomized response (RR)} mechanism~\cite{warner1965randomized}. Specifically, $\RR_\eps$ is an algorithm that takes in $x \in \{0, 1\}^n$ and outputs $\tx \in \{0, 1\}^n$, where $\tx_i = x_i$ with probability $\frac{e^\eps}{1 + e^\eps}$ independently for each $i \in [n]$. It is well-known (and very simple to verify) that $\RR_\eps$ is $\eps$-$\SDP$.\\[-3mm]

\myparagraph{Collision-Resistant Keyless Hash Functions.}
In our construction, we will use the Collision-Resistant Keyless Hash Functions (CRKHFs)~\cite{BitanskyKP18}. The formal definition is as given below.

\begin{definition}[Collision-Resistant Keyless Hash Functions~\cite{BitanskyKP18}]
A sequence of hash functions $\seq{H_{n} : \bit^{n} \to \bit^{\gamma(n)}}_{n \in \N}$ is {\em $K$-collision resistant} for advice length $\zeta$ for sequences $K = \set{K_n}_{n \in \N}$, $\zeta = \set{\zeta_n}_{n \in \N}$ if, for any PPT $\cA$ and a sequence  $\{z_n\}_{n \in \N}$ of advices where $|z_n| = \zeta_n$, it holds for $(Y_1, \dots, Y_{K_n}) \gets \cA(1^{n}; z_n)$ that
\begin{align*}
\Pr\begin{bmatrix}
    Y_1, \dots, Y_{K_n} \text{ are distinct \&}~ \\
    H_{n}(Y_1) = \cdots = H_{n}(Y_{K_n})
\end{bmatrix} & \leq \negl(n).
\end{align*}
We skip the subscript $n$ when it is clear from context.
\end{definition}

In~\cite{BitanskyKP18}, the hash value length $\gamma(n)$ is assumed to be either linear, i.e., $\gamma(n) = \Omega(n)$, or polynomial, i.e., $\gamma(n) = n^{\Theta(1)}$. However, we need a collision-resistant hash function with a much smaller $\gamma(n)$, namely $O(\log^2 n)$. We remark that this is still very much plausible: as long as $\gamma(n)$ is $\omega(\log n)$, the ``guess-and-check'' algorithm will only produce a collision with only negligible probability.
A more precise statement of our assumption is stated below.

\begin{assumption} \label{as:crklhf}
There is an efficiently computable sequence $H = \set{H_n}_{n \in \N}$ of hash functions with hash value length $\gamma(n) = o(\log^2 n)$ such that, for any constant $c_1 > 0$, there exists a constant $c_2 > 0$ such that the hash function sequence is $K$-collision resistant for advice length $\zeta$ where $K_n = n^{c_2}$ and $\zeta_n = n^{c_1}$.
\end{assumption}

We remark that, for the existence of $\CDP$ mechanism (shown in this section), we will only use the multi-collision-resistance without relying on the assumption on the value of $\gamma$. The latter is only used to show that no $\SDP$ mechanism exists for the problem
(\Cref{sec:proof-main-thm}).

\paragraph{Differing-Inputs Obfuscators for Public-Coin Samplers.}
For any two circuits $C_0$ and $C_1$, a differing-inputs obfuscator $\diO$ \cite{BarakGIRSVY12} guarantees that the non-existence of an efficient adversary that can find an input on which $C_0$ and $C_1$ differ implies that $\diO(C_0)$ and $\diO(C_1)$ are computationally indistinguishable. For our application, it suffices to assume a weaker notion, namely that of {\em differing-inputs obfuscator for public-coin samplers}, as defined below.

\begin{definition}[Public-Coin Differing-Inputs Circuit Sampler]\label{def:pcSamplers}
An efficient non-uniform sampling algorithm $\Sampler = \set{\Sampler_n}$ is a \emph{public-coin differing-inputs sampler} for the parameterized collection $\cC = \set{\cC_n}$ of circuits if the output of $\Sampler_n$ is distributed over $\cC_n \times \cC_n$ and for every efficient non-uniform algorithm $\cA = \set{\cA_n}$, there exists a negligible function $\negl(\cdot)$ such that for all $n \in \N$:
\begin{align*}
\Pr_{\theta} \begin{bmatrix*}[l]
C_0(y) \ne C_1(y) :\\
(C_0, C_1) \gets \Sampler_n(\theta),\\
y \gets \cA_n(\theta)
\end{bmatrix*}  & ~\le~ \negl(n).
\end{align*}
Here, $\Sampler_n$ is a deterministic algorithm and the only source of randomness is the seed $\theta$.
\end{definition}

\begin{definition}[Differing-Inputs Obfuscator for Public-Coin Samplers (cf. \cite{IshaiPS15})]\label{def:diO}
A uniform PPT $\diO$ is a \emph{differing-inputs obfuscator for public-coin samplers} for the parameterized circuit family $\cC = \{\cC_{n}\}$ if the following conditions are satisfied:
\begin{itemize}[nosep,leftmargin=*]
\item {\bf Correctness:} For all $n \in \mathbb{N}$, for all $C \in \cC_n$, for all inputs $y$, we have that
\[ \Pr[C'(y) = C(y) : C' \gets \diO(1^n, C)] ~=~ 1. \]
\item {\bf Polynomial slowdown:} There exists a universal polynomial $p(\cdot)$ such that for all $C \in \cC_n$, it holds that \[ \Pr[|C'| \le p(|C|) : C' \gets \diO(1^n, C)] ~=~ 1. \]
\item {\bf Differing-inputs:} For every public-coin differing inputs sampler $\Sampler = \set{\Sampler_n}$ for $\cC$, and every (not necessarily uniform) PPT distinguisher $\cD = \set{D_n}$, there exists a negligible function $\negl$ such that the following holds for all $n \in \N$: For $(C_0, C_1) \gets \Sampler_n(\theta)$
\begin{align*}
\begin{vmatrix*}[l]
\Pr_{\theta} \left [D_n(\diO(1^n, C_0)) = 1 \right] \\
- \Pr_{\theta} \left [D_n(\diO(1^n, C_1)) = 1 \right ]
\end{vmatrix*}~\le~ \negl(n).
\end{align*}
\end{itemize}
We note that the notion of $\diOpcS$ is in fact weaker than the notion of general public-coin $\diO$ as given by \cite{IshaiPS15}. We elaborate on this comparison in Appendix~\ref{apx:diO-comparison}. Whenever $n$ is clear from context, we use $\diO(C)$ to denote $\diO(1^n, C)$ for simplicity. When we want to be explicit about the randomness $\rho$ (of $\poly(n)$ bit length) used by $\diO$ we will denote it as $\diO_{\rho}(C)$.
\end{definition}

\noindent We only need the existence of a differing-inputs obfuscator for a specific family of circuits. This circuit family will be defined later and therefore we defer formalizing our assumption to \Cref{subsec:cdp-unverifiable}.

\subsubsection{Public-Coin Differing-Inputs Circuits from CRKHFs}
\label{subsec:di-circuit-family}

The first step of our proof is to construct a differing-inputs circuit family based on CRKHFs.
Our sampler is described in \Cref{alg:di-sampler}.

\begin{algorithm}[ht]
\caption{Differing-Inputs Circuit Family Sampler $\LDSSampler_n$.}
\label{alg:di-sampler}
\begin{algorithmic}
\STATE {\bf Parameters: } Adjacent datasets $x, x' \in \bit^n$, hash value $\upsilon_n \in \bit^{\gamma(n)}$, privacy parameter $\eps > 0$, radius $r, \tr > 0$.
\STATE {\bf Randomness:} $\theta \sim \RR_\eps(0^n)$.
\STATE {\bf Output:} Circuits $C_0, C_1$.
\STATE $\tx \gets x \oplus \theta$ (bit-wise XOR; equivalent to $\RR_{\eps}(x)$)
\STATE $C_0 \gets$ circuit that on input $z$ returns $\indicator\set{z \in B_r(x) \cap B_{\tr}(\tilde{x}) \cap H^{-1}_{n}(\upsilon_n)}$
\STATE $C_1 \gets$ circuit that on input $z$ returns $\indicator\set{z \in B_r(x') \cap B_{\tr}(\tilde{x}) \cap H^{-1}_{n}(\upsilon_n)}$
\RETURN $(C_0, C_1)$
\end{algorithmic}
\end{algorithm}

We next prove that the above sampler is a public-coin differing-inputs sampler, which means that any efficient adversary, even with the knowledge of $\tx$ (which is the only source of randomness), cannot find an input on which $C_0$ and $C_1$ differ. The proof starts by noticing that any input that differentiates $C_0, C_1$ must, by definition of the circuits, have hash value $\upsilon_n$. Therefore, if there were an adversary that can find a differing input, then we could run it multiple times to get $Y_1, \dots, Y_K$ that have the same hash value. (See \Cref{adv:crhf} below.) However, our proof is not finished yet, since it is possible that $Y_1, \dots, Y_K$ are not distinct. Indeed, the crux of the construction is that, due to how we select $\tx$ and define the circuits, a fixed $Y$ will be a differing input with negligible probability\footnote{It is also simple to see that this property suffices to prove a non-adaptive query lower bound as discussed in \Cref{sec:proof-overview}.}. It follows that $Y_1, \dots, Y_K$ must be distinct w.h.p. This is formalized below.

\begin{lemma} \label{lem:di-circuit-family}
Let $H$ be as in \Cref{as:crklhf}. For any constant $\eps > 0$, choosing $r = 0.5n^{0.9}$ and $\tr = \frac{1}{1 + e^\eps} n + n^{0.6}$ makes $\LDSSampler_n$ (\Cref{alg:di-sampler}) a public-coin differing-inputs sampler.
\end{lemma}

\begin{proof}
Suppose for the sake of contradiction that for some adjacent $x, x' \in \bit^n$, there exists a PPT $\ADI$ such that
\begin{align} \label{eq:di-violation}
\Pr_{\theta} \begin{bmatrix*}[l]
C_0(y) \ne C_1(y) :\\
(C_0, C_1) \gets \LDSSampler_n(\theta),\\
y \gets \ADI_n(\theta)
\end{bmatrix*}  & ~\ge~ n^{-c}.
\end{align}
for some constant $c > 0$. Furthermore, let $c_1$ be such that the total size of the descriptions of $\ADI_n, \LDSSampler_n$ is at most $n^{c_1}$. Finally, let $c_2 > 0$ be as in \Cref{as:crklhf} and $K = n^{c_2}$.

\begin{algorithm}[ht]
\caption{Collision-Resistant Hash Function Adversary $\ACRH_n$.}
\label{adv:crhf}
\begin{algorithmic}
\STATE {\bf Parameter: } The target number of collisions $K \in \N$, constant $c > 0$.
\STATE {\bf Advice: } Descriptions of $\ADI_n, \LDSSampler_n$.
\STATE {\bf Output:} $Y_1, \dots, Y_K \in \{0, 1\}^n$ or $\perp$.
\STATE $i \gets 0$
\FOR{$j = 1, \dots, K \cdot n^{c+1}$}
\STATE $\theta^j \gets \RR_\eps(0^n)$
\STATE $(C^j_0, C^j_1) \gets \LDSSampler_n(\theta^j)$
\STATE $y^j \gets \ADI_n(\theta^j)$
\IF{$C^j_0(y^j) \ne C^j_1(y^j)$}
\STATE $i \gets i + 1$
\STATE $Y_i \gets y^j$
\ENDIF
\IF{$i \geq K$}
\STATE \textbf{break}
\ENDIF
\ENDFOR
\IF{$i < K$}
\RETURN $\perp$
\ELSE
\RETURN $Y_1, \dots, Y_K$
\ENDIF
\end{algorithmic}
\end{algorithm}

Consider the adversary $\ACRH_n$ for collision-resistant hash function described in \Cref{adv:crhf}.
First, note that by \eqref{eq:di-violation} and a standard concentration inequality, the probability that $\ACRH_n$ outputs $\perp$ is $o_n(1)$. Furthermore, notice that $C_0, C_1$ can differ on $y$ only if $H_n(y) = \upsilon_n$, meaning that $H_n(Y_i) = \upsilon_n$ always. Therefore, it suffices for us to show that the probability that $Y_1, \dots, Y_K$ are distinct is $1 - o_n(1)$. By a union bound, we have that $\ACRH_n$ violates the collision-resistance of $H$ as desired.

Thus, we are only left to show that $Y_1, \dots, Y_K$ are not distinct with probability $o(1)$. To see that this is the case, notice that
\begin{align}
&\Pr[Y_1, \dots, Y_K \text{ are not distinct}] \nonumber\\
&~\leq~\sum_{1 \leq i_1 < i_2 \leq K} \Pr[Y_{i_1} = Y_{i_2}].\label{eq:prob-distinct-expand}
\end{align}
Let us now bound $\Pr[Y_{i_1} = Y_{i_2}]$ for a fixed pair $i_1 < i_2$. Suppose that we fix a value of $Y_{i_1}$ and suppose that $Y_{i_1}$ is assigned at step $j_1 \in [1, \dots, K \cdot n^{c+1}]$. Conditioned on these, notice further that
\begin{align}
\Pr[Y_{i_2} = Y_{i_1}] &\leq \Pr[\exists j > j_1, y^j = Y_{i_1}] \nonumber \\
&\leq \Pr[\exists j > j_1, C^j_0(Y_{i_1}) \ne C^j_1(Y_{i_1})] \nonumber \\
&\leq \sum_{j > j_1} \Pr[C^j_0(Y_{i_1}) \ne C^j_1(Y_{i_1})]. \label{eq:prob-distinct-expand2}
\end{align}
Now, let us bound the RHS probability for a fixed $j > j_1$. To see this, first observe that $Y_{i_1}$ must belong to the symmetric difference $B_r(x) \triangle B_r(x')$; otherwise, we must have $C^{j_1}_0(Y_{i_1}) = C^{j_1}_1(Y_{i_1})$, a contradiction to our definition of $Y_{i_1}$.

\noindent Now, let $\tx^j$ denote the $\tx$ selected by $\LDSSampler$ when constructing $C^j_0, C^j_1$. We have
\begin{align} \label{eq:diff-prob-each}
\Pr[C^j_0(Y_{i_1}) \ne C^j_1(Y_{i_1})] \leq \Pr[Y_{i_1} \in B_{\tr}(\tx^j)].
\end{align}
Let $d := \|Y_{i_1} - x\|_1$ and $\td := \|Y_{i_1} - \tx^j\|_1$. 
Since $Y_{i_1} \in B_r(x) \triangle B_r(x')$, it holds that $d \in \{r, r + 1\}$. Thus, $\td$ is distributed as $\Bin(d, \frac{e^{\eps}}{1+e^{\eps}}) + \Bin(n-d, \frac{1}{1+e^{\eps}})$. We have $\E_{\tx^j \sim \RR_{\eps}(x)} \td = \frac{1}{1+e^{\eps}} n + \frac{e^{\eps} - 1}{e^{\eps} + 1} d$. By Bernstein's inequality,
\begin{align*}
    \Pr[\td \leq \tr] & ~\le~ \exp\inparen{- \frac{t^2}{\frac{e^{\eps}}{(1+e^{\eps})^2} n + \frac{2}{3} t}} \\
    & ~\le~ \exp(-\Omega(n^{0.8})),
\end{align*}
where $t = \E_{\tx^j \sim \RR_{\eps}(x)} \td - \tr \geq \frac{e^{\eps} - 1}{e^{\eps} + 1}(0.5n^{0.9} - 1) - n^{0.6}$. Plugging into \eqref{eq:diff-prob-each}, we have %
\begin{align} \label{eq:diff-prob-each-concrete}
\Pr[C^j_0(Y_{i_1}) \ne C^j_1(Y_{i_1})] \leq \exp(-\Omega(n^{0.8})).
\end{align}
Combing \eqref{eq:prob-distinct-expand}, \eqref{eq:prob-distinct-expand2}, \eqref{eq:diff-prob-each-concrete}, we have
\begin{align*}
&\Pr[Y_1, \dots, Y_K \text{ are not distinct}]\\
& ~\leq~ K^3 n^{c+1} \cdot \exp(-\Omega(n^{0.8}))\\
& ~\le~ \exp(-\Omega(n^{0.8})),
\end{align*}
where the last inequality follows from $K = n^{O(1)}$.
\end{proof}

\subsubsection{\boldmath From Differing-Inputs Circuits to \texorpdfstring{$\CDP$}{CDP}}
\label{subsec:cdp-unverifiable}

We will next construct a $\CDP$ mechanism from the previously constructed differing-inputs circuit family. First, let us state the assumption we need here:

\begin{assumption} \label{as:diO}
For $H$ as in \Cref{as:crklhf}, any constant $\eps > 0$ and $r = 0.5n^{0.9}, \tr = \frac{1}{1 + e^\eps} n + n^{0.6}$, there exists a differing-inputs obfuscator $\diO$ for the sampler $\LDSSampler$.
\end{assumption}

Our mechanism can then be defined by simply applying the obfuscator to the circuit generated in the same way as $C_1$ in $\LDSSampler_n$. This mechanism $\Mdio$ is described more formally in \Cref{alg:dio}. The $\CDP$ property of the mechanism follows rather simply from the definition of $\diO$ and the fact that $\RR_{\eps}$ is $\eps$-$\SDP$. 

\begin{algorithm}[ht]
\caption{$\CDP$ mechanism $\Mdio$.}
\label{alg:dio}
\begin{algorithmic}
\STATE {\bf Parameter: } Differing-inputs obfuscator $\diO$, hash function $H$, parameters $\eps, r, \tr$ (as in \Cref{as:diO}), and a hash value $\upsilon_n \in \{0, 1\}^{\gamma(n)}$.
\STATE {\bf Input:} Dataset $x \in \bit^n$.
\STATE {\bf Output:} Circuit $: \{0, 1\}^n \to \{0, 1\}$.
\STATE $\tx \gets \RR_{\eps}(x)$.
\STATE $C \gets$ circuit that on input $z$ returns $\indicator\set{z \in B_r(x) \cap B_{\tr}(\tx) \cap H^{-1}_{n}(\upsilon_n)}$
\STATE $\hC \gets \diO_{\rho}(C)$ for randomness $\rho$
\RETURN $\hC$
\end{algorithmic}
\end{algorithm}

\begin{figure*}[t]
\centering
\begin{tikzpicture}
\tikzset{
    hybrid/.style = {draw, text width=7.2cm, rounded corners=2pt, line width=0.7pt, inner sep=4pt}
}
\node[hybrid] at (-3.9,0) {
	\textbf{\boldmath Distribution $H_0$:}\\
	\begin{algorithmic}
	\STATE $\tx \gets \RR_{\eps}(x)$
	\STATE $C(z) := \indicator\set{z \in B_r(x) \cap B_{\tr}(\tx) \cap H^{-1}_{n}(\upsilon_n)}$
	\RETURN $\diO_\rho(C)$
	\end{algorithmic}
};
\node[hybrid] at (3.9,0) {
	\textbf{\boldmath Distribution $H_1$:}\\
	\begin{algorithmic}
	\STATE $\tx \gets \RR_{\eps}(x)$
	\STATE $C(z) := \indicator\set{z \in B_r(\textcolor{red}{x'}) \cap B_{\tr}(\tx) \cap H^{-1}_{n}(\upsilon_n)}$
	\RETURN $\diO_\rho(C)$
	\end{algorithmic}
};
\node[hybrid] at (0,-2.3) {
    \textbf{\boldmath Distribution $H_2$:}\\
    \begin{algorithmic}
	\STATE $\tx \gets \RR_{\eps}(\textcolor{red}{x'})$
	\STATE $C(z) := \indicator\set{z \in B_r(x') \cap B_{\tr}(\tx) \cap H^{-1}_{n}(\upsilon_n)}$
	\RETURN $\diO_\rho(C)$
    \end{algorithmic}
};
\end{tikzpicture}
\caption{Hybrids in proof of \Cref{thm:Mdio}. $H_0$ is precisely $\Mdio(x)$ and $H_2$ is precisely $\Mdio(x')$.}
\label{fig:Mdio-hybrids}
\end{figure*}

\begin{theorem}\label{thm:Mdio}
Under \cref{as:crklhf,as:diO}, $\Mdio$ (\Cref{alg:dio}) is $\eps$-$\CDP$.
\end{theorem}

\begin{proof}
For any adjacent datasets $x, x'$, we want to show that $\Mdio(x) \approx^c_{\eps} \Mdio(x')$. We show this using an intermediate hybrid, as shown in \cref{fig:Mdio-hybrids}, where changes from one hybrid to next are highlighted in \textcolor{red}{red}.

\begin{itemize}[leftmargin=*]
    \item Distribution $H_0$ is precisely $\Mdio(x)$.
    \item Distribution $H_1$ is a variant of $H_0$, where we change $x$ to $x'$ in the definition of $C$, but continue to sample $\tx \sim \RR_{\eps}(x)$.
    \item Distribution $H_2$ is a variant of $H_1$, where we sample $\tx \sim \RR_{\eps}(x')$. Note that this is exactly $\Mdio(x')$.
\end{itemize}

We show that $H_0 \approx^c_{\eps} H_2$ by showing that $H_0 \approx^c H_1$ and $H_1 \approx_{\eps, 0} H_2$ and using basic composition (\Cref{fact:dp-calculus}).
We have from \Cref{lem:di-circuit-family}, that under \Cref{as:crklhf}, the joint distribution of $\tx \sim \RR_{\eps}(x)$, and circuits $C$ in $H_0$ and $H_1$ is precisely the output of $\LDSSampler$. Thus, from \Cref{as:diO}, it follows that $H_0 \approx^c H_1$ by post-processing~(\Cref{fact:dp-calculus}).
Next, we have that $H_1 \approx_{(\eps, 0)} H_2$, since the only difference between the two is the distribution of $\tx$, and $\RR_{\eps}(x) \approx_{(\eps,0)} \RR_{\eps}(x')$ (again by post-processing).
\end{proof}

Finally, its utility also follows simply from a standard concentration inequality.

\begin{theorem} \label{thm:mdio-util}
When choosing $\tr = \frac{1}{1 + e^\eps} n + n^{0.6}$, $\Mdio$ is $(1 - o(1))$-useful for $u^{\eval}_{H^{-1}_n(\upsilon_n)}$.
\end{theorem}

\begin{proof}
Consider any dataset $x$. If $x \notin H^{-1}_n(\upsilon_n)$, then, by definition of $u^{\LDS}_{H^{-1}_n(\upsilon_n)}$, the utility is $1$. Therefore, we may only consider the case where $x \in H^{-1}_n(\upsilon_n)$.

In this case, $\Pr\left[u^{\eval}_{H^{-1}_n(\upsilon_n)}(x, \Mdio(x)) = 1\right]$ is equal to $\Pr_{\tx \sim \RR_\eps(x)}[x \in B_{\tr}(\tx)]$. Notice that $\|x - \tx\|_1$ is distributed as $\Bin(n, \frac{1}{1+e^\eps})$. Therefore, applying Bernstein’s inequality, we have
\begin{align*}
\Pr_{\tx \sim \RR_\eps(x)}[x \notin B_{\tr}(\tx)] &~\le~ \exp\inparen{- \frac{t^2}{\frac{e^{\eps}}{(1+e^{\eps})^2} n + \frac{2}{3} t}} \\
& ~\le~ \exp(-\Omega(n^{0.2})),
\end{align*}
where $t = \tr - \frac{n}{1+e^\eps} = n^{0.6}$. Thus, we have $\Pr\left[u^{\eval}_{H^{-1}_n(\upsilon_n)}(x, \Mdio(x)) = 1\right] = 1 - o(1)$ as desired.
\end{proof}

\subsection{\texorpdfstring{$\CDP$}{CDP} Mechanism for \texorpdfstring{$\VLDS$}{VLDS}}
\label{sec:cdp-verifiable}

\subsubsection{Witness-Indistinguishable Proofs}
For any $\NP$ language $L$ with associated verifier $V_L$, let $R_L$ denote the corresponding relation $\set{(x,w) : x \in L \text{ and } V_L(x,w) = 1}$. Let $R_L(x) := \set{w : (x,w) \in R_L}$.

\begin{definition}[NIWI Proof System]
A pair $(P, V)$ of PPT algorithms is a {\em non-interactive witness indistinguishable (NIWI) proof system} for an $\NP$ relation $R_L$ if it satisfies:%
\begin{description}[leftmargin=3mm]
\item\textbf{Correctness:} for every $(x, w) \in R_L$
\[ \Pr[V(x,\pi) = 1 : \pi \gets P(x, w)] = 1. \]
\item\textbf{Soundness:} there exists a negligible function $\negl$ such that for all $x \notin L$ and $\pi \in \bit^*$:
\[ \Pr[V(x,\pi) = 1] \le \negl(|x|). \]
\item\textbf{Witness Indistinguishability:} There exists a polynomial $\zeta(\cdot)$ and a negligible function $\negl(\cdot)$, such that for any sequence $I = \set{(x, w_0, w_1) : w_0, w_1 \in R_L(x)}$ and for all circuits $C$ of size at most $\zeta(|x|)$:
\begin{align*}
\begin{vmatrix*}[l]
\Pr_{\pi_0 \gets P(x, w_0)}[C(x, \pi_0) = 1] \\
- \Pr_{\pi_1 \gets P(x, w_1)}[C(x, \pi_1) = 1]
\end{vmatrix*}&~\le \negl(|x|). 
\end{align*}
\end{description}
\end{definition}

\begin{assumption}[\cite{BarakOV07,GrothOS12,BitanskyP15}]\label{as:niwi}
There exists a NIWI proof system for any language in $\NP$.
\end{assumption}

\subsubsection{Making Utility Function Efficient Using Witness-Indistinguishable Proofs}\label{sec:NIWI}

We consider the $\NP$ language $\hL$ defined below, and use the corresponding NIWI verifier to define the utility for $\VLDS$.

\begin{definition}\label{def:niwi-language}
Language $\hL$ consists of all circuits $\hC$ with a top AND gate, namely of the form $\hC_0 \wedge \hC_1$ such that there exists some $x$, $\tx$ and $\rho$, such that at least one of $\hC_0$ or $\hC_1$ can be obtained as $\diO_\rho(C)$ where $C$ is a circuit that takes in $z$ and computes $\indicator\set{z \in B_r(x) \cap B_{\tr}(\tx) \cap H^{-1}(\upsilon)}$.

A ``witness'' for $\hC \in \hL$ is given by $w = (b, x, \tx, \rho)$, where $b \in \bit$ indicates whether the witness is provided for $\hC_0$ or for $\hC_1$. Let $(\hP, \hV)$ denote the NIWI proof system for $L$ (guaranteed to exist by \Cref{as:niwi}).
\end{definition}

We consider the verifiable low diameter set problem $\VLDS_{\tau, H^{-1}(\upsilon), \hV}$. Note that $\hC \in \hL$ automatically implies that $\hC$ encodes a $\tau$-diameter set (since $\hC = \hC_0 \wedge \hC_1$, it suffices to certify that at least one of $\hC_0$ or $\hC_1$ encodes a $\tau$-diameter set) where $\tau = 2r = n^{0.9}$.

\begin{algorithm}[t]
\caption{Sub-routine $\Mdioaux$.}
\label{alg:dioaux}
\begin{algorithmic}
\STATE {\bf Parameter: } Differing-inputs obfuscator $\diO$, hash function $H$, parameters $\eps, r, \tr$ (as in \Cref{as:diO}), and a hash value $\upsilon \in \{0, 1\}^{\gamma(n)}$.
\STATE {\bf Input:} Dataset $x \in \bit^n$.
\STATE {\bf Output:} Circuit $: \bit^n \to \bit$.
\STATE $\tx \gets \RR_{\eps}(x)$.
\STATE $C \gets$ circuit that on input $z$ returns $\indicator\set{z \in B_r(x) \cap B_{\tr}(\tx) \cap H^{-1}_{n}(\upsilon)}$
\STATE $\hC \gets \diO_{\rho}(C)$ for randomness $\rho$
\RETURN $\hC$, $\tx$, $\rho$
\end{algorithmic}
\end{algorithm}

\begin{algorithm}[t]
\caption{$\CDP$ mechanism $\Mcdp$.}
\label{alg:Mcdp}
\begin{algorithmic}
\STATE {\bf Input:} Dataset $x \in \bit^n$, radius parameters $r, \tr > 0$ and privacy parameter $\eps$.
\STATE {\bf Output:} Circuit $C$ and a proof string $\pi$.
\STATE $\hC_0, \tx_0, \rho_0 \gets \Mdioaux(x)$
\STATE $\hC_1, \tx_1, \rho_1 \gets \Mdioaux(x)$
\STATE $\hC = \hC_0 \wedge \hC_1$
\STATE $\pi \gets \hP(\hC, (0, x, \tx_0, \rho_0))$ (NIWI proof for $\hC \in \hL$ using witness $(0, x, \tx_0, \rho_0)$).
\RETURN $\hC$, $\pi$
\end{algorithmic}
\end{algorithm}

\begin{theorem} \label{thm:final-cdp}
Under \cref{as:crklhf,as:diO,as:niwi}, $\Mcdp$ (\cref{alg:Mcdp}) is $2\eps$-$\CDP$.
\end{theorem}

\begin{proof}
For any adjacent datasets $x$, $x'$, we want to show that $\Mcdp(x) \approx^c_{2\eps} \Mcdp(x')$. We show this through the means of intermediate hybrids, as shown in \Cref{fig:Mcdp-hybrids}, where changes from one hybrid to next are highlighted in \textcolor{red}{red}.

\begin{figure*}
\centering
\begin{tikzpicture}
\tikzset{
    hybrid/.style = {draw, text width=4.3cm, rounded corners=2pt, line width=0.7pt, inner sep=4pt}
}
\node[hybrid] at (-5,0) {
	\textbf{\boldmath Distribution $H_0$:}\\
	\begin{algorithmic}
	\STATE $\hC_0, \tx_0, \rho_0 \gets \Mdioaux(x)$
	\STATE $\hC_1, \tx_1, \rho_1 \gets \Mdioaux(x)$
	\STATE $\hC = \hC_0 \wedge \hC_1$
	\STATE $\pi \gets \hP(\hC, (0, x, \tx_0, \rho_0))$
	\RETURN $\hC$, $\pi$
	\end{algorithmic}
};
\node[hybrid] at (0,0) {
	\textbf{\boldmath Distribution $H_1$:}\\
	\begin{algorithmic}
	\STATE $\hC_0, \tx_0, \rho_0 \gets \Mdioaux(x)$
	\STATE $\hC_1, \tx_1, \rho_1 \gets \Mdioaux(\textcolor{red}{x'})$
	\STATE $\hC = \hC_0 \wedge \hC_1$
	\STATE $\pi \gets \hP(\hC, (0, x, \tx_0, \rho_0))$
	\RETURN $\hC$, $\pi$
	\end{algorithmic}
};
\node[hybrid] at (5,0) {
    \textbf{\boldmath Distribution $H_2$:}\\
    \begin{algorithmic}
	\STATE $\hC_0, \tx_0, \rho_0 \gets \Mdioaux(x)$
	\STATE $\hC_1, \tx_1, \rho_1 \gets \Mdioaux(x')$
	\STATE $\hC = \hC_0 \wedge \hC_1$
	\STATE $\pi \gets \hP(\hC, \textcolor{red}{(1, x', \tx_1, \rho_1)})$.
	\RETURN $\hC$, $\pi$
    \end{algorithmic}
};
\node[hybrid] at (-2.5,-3.2) {
    \textbf{\boldmath Distribution $H_3$:}\\
    \begin{algorithmic}
	\STATE $\hC_0, \tx_0, \rho_0 \gets \Mdioaux(\textcolor{red}{x'})$
	\STATE $\hC_1, \tx_1, \rho_1 \gets \Mdioaux(x')$
	\STATE $\hC = \hC_0 \wedge \hC_1$
	\STATE $\pi \gets \hP(\hC, (1, x', \tx_1, \rho_1))$.
	\RETURN $\hC$, $\pi$
    \end{algorithmic}
};
\node[hybrid] at (2.5,-3.2) {
    \textbf{\boldmath Distribution $H_4$:}\\
    \begin{algorithmic}
	\STATE $\hC_0, \tx_0, \rho_0 \gets \Mdioaux(x')$
	\STATE $\hC_1, \tx_1, \rho_1 \gets \Mdioaux(x')$
	\STATE $\hC = \hC_0 \wedge \hC_1$
	\STATE $\pi \gets \hP(\hC, \textcolor{red}{(0, x', \tx_0, \rho_0)})$
	\RETURN $\hC$, $\pi$
    \end{algorithmic}
};
\end{tikzpicture}
\caption{Hybrids in proof of \Cref{thm:final-cdp}. $H_0$ is precisely $\Mcdp(x)$ and $H_4$ is precisely $\Mcdp(x')$.}
\label{fig:Mcdp-hybrids}
\end{figure*}

\begin{itemize}[nosep,leftmargin=*]
\item Distribution $H_0$ is precisely $\Mcdp(x)$.

\item Distribution $H_1$ is a variant of $H_0$, where $\hC_1$ is generated through $x'$ instead of $x$.

\item Distribution $H_2$ is a variant of $H_1$, where we switch $\pi$ from corresponding to witness $(0, x, \tx_0, \rho_0)$ to the witness $(1, x', \tx_1, \rho_1)$.

\item Distribution $H_3$ is a variant of $H_2$, where $\hC_0$ is also generated through $x'$ instead of $x$.

\item Distribution $H_4$ is a variant of $H_3$, where we switch $\pi$ from corresponding to witness $(1,x', \tx_1, \rho_1)$ to the witness $(0, x', \tx_0, \rho_0)$. Note that this is exactly $\Mcdp(x')$.
\end{itemize}

\noindent From \cref{as:niwi} and post-processing (\Cref{fact:dp-calculus}), we have that $H_1 \approx^c H_2$, and similarly $H_3 \approx^c H_4$.

Next, we show that $H_0 \approx^c_{\eps} H_1$. Note that the output of $H_0$ and $H_1$ do not depend on $\tx_1$ and $\rho_1$. Thus the only material change between $H_0$ and $H_1$ is that $\hC_1 \sim \Mdio(x)$ in $H_0$ versus $\hC_1 \sim \Mdio(x')$ in $H_1$. From \cref{thm:Mdio}, we have that $\Mdio(x) \approx^c_\eps \Mdio(x')$. Thus, it follows that $H_0 \approx^c_{\eps} H_1$ by post-processing~(\Cref{fact:dp-calculus}).
Similarly, it follows that $H_2 \approx^c_\eps H_3$ (here we use that $\tx_0$ and $\rho_0$ are immaterial to the final output of $H_2$ and $H_3$).

Combining these using basic composition (\Cref{fact:dp-calculus}), we get that $H_0 \approx^c_{2\eps} H_4$, thus implying that $\Mcdp$ is $2\eps$-$\CDP$.
\end{proof}

\begin{corollary} \label{cor:mcdp-useful}
$\Mcdp$ is $(1 - o(1))$-useful for $u^{\VLDS}_{\tau, H^{-1}(\upsilon), \hV}$.
\end{corollary}

\begin{proof}
The utility for $x \notin H^{-1}(\upsilon)$ is trivially $1$. Consider $x \in H^{-1}(\upsilon)$. Suppose the mechanism $\Mdio$ is $(1 - \eta)$-useful for $u^{\eval}_{H^{-1}(\upsilon)}$. Since we sample $\hC_0$ and $\hC_1$ from $\Mdio$ independently we have that $\hC(x) = 1$ with probability at least $1-2\eta$. Finally, note that the proof $\pi$ in the output of $\Mcdp$ is always accepted by $\hV$. From \cref{thm:mdio-util}, we have that $\eta = o(1)$, and hence $\Mcdp$ is $1-2\eta = 1-o(1)$ useful for $u^{\VLDS}_{\tau, H^{-1}(\upsilon), \hV}$.
\end{proof}

We end this section by proving \Cref{thm:cdp-main}. The proof is essentially a straightforward combination of the previous two results. The only choice left to make is to select the hash value $\upsilon$; we select it so that the size of the preimage $H^{-1}(\upsilon)$ is maximized. This ensures that the set $\cR = H^{-1}(\upsilon)$ has enough density as required in \Cref{thm:cdp-main}. (Note: the density requirement in \Cref{thm:cdp-main} is not important for showing the existence of a $\CDP$ mechanism, but instead is later used to show the non-existence of $\SDP$ mechanisms.) 

\begin{proof}[Proof of \Cref{thm:cdp-main}]
Let $H, \tau, \hV$ be as defined above. Furthermore, let $\upsilon$ be such that $H^{-1}(\upsilon)$ is maximized and $\eps = \eps_{\CDP} / 2$. The fact that there exists an $\eps_{\CDP}$-$\CDP$ mechanism that is $(1 - o(1))$-useful for $u^{\VLDS}_{\tau, \cR, \hV}$ follows immediately from \Cref{thm:final-cdp} and \Cref{cor:mcdp-useful}. Furthermore, by our choice of $\upsilon$, notice that $|\cR| = |H^{-1}(\upsilon)| \geq 2^{n} / 2^{\gamma(n)} \geq 2^{n} / n^{o(\log n)}$, where the latter comes from our assumption on $\gamma$ in \Cref{as:crklhf}. 
\end{proof}%
\section{\texorpdfstring{$\SDP$}{SDP} Lower Bounds for the Nearby Point Problem}
\label{sec:SDP}

In this section, we will show that there is no $O(1)$-$\SDP$ algorithm for the nearby point problem with target threshold $n^{0.9}$ as long as the set $\cR_n$ is fairly dense, as formalized below.

\begin{theorem} \label{thm:no-dp-main}
For $\tau = \seq{\tau_n}_{n \in \N}$ and $\cR = \seq{\cR_n \subseteq \bit^n}_{n \in \N}$ such that $\tau_n \leq n^{0.9}$ and $|\cR_n| \geq 2^n / n^{o(\log n)}$ and for any constant $\beps, \balpha > 0$ and $\bdelta = 1/n^{27}$, no $(\beps, \bdelta)$-$\SDP$ mechanism is $\balpha$-useful for $u^{\NBP}_{\tau, \cR}$.
\end{theorem}

To prove \Cref{thm:no-dp-main}, let us first recall the standard ``blatant non-privacy implies non-DP'' proof\footnote{Here we follow the proofs in \cite{Suresh19,Manurangsi22}.}, which corresponds to the case $\cR_n = \bit^n$. At a high-level, these proofs proceed by showing that the error in each coordinate is large by ``matching'' each $x \in \bit^n$ with another point $x'$ which is the same as $x$ except with the $i$-th bit flipped; a basic calculation then shows that (on average) the $i$-th bit is predicted incorrectly with large probability. Summing this up over all the coordinates yield the desired bound.

As we are in the case where $\cR_n \ne \bit^n$, we cannot use the proof above directly. Nonetheless, we can still adapt the above proof. More specifically, instead of looking at each coordinate at a time, we look at a block of coordinates. For each block, we try to find a matching in the same spirit as above, but we now allow the $x, x'$ to have a larger distance; simple calculations give us a lower bound on being incorrect in this block (\Cref{sec:sdp-each-block}). We then ``sum up'' across all blocks to get a large distance (\Cref{sec:boost-dist}). Even though we get a large distance $\tau$ via this approach, the error probability (i.e. one minus usefulness) is small (i.e. $o(1)$). Fortunately, we can overcome this using the so-called DP hyperparameter tuning algorithm~\cite{LT19,PS21} (\Cref{sec:boost}). This concludes our proof overview.

\subsection{Additional Preliminaries: Tools from Differential Privacy}

We will require several additional tools from DP literature, which we list below for completeness.

\myparagraph{Laplace Mechanism.} The \emph{Laplace distribution} with scale parameter $b > 0$, denoted by $\Lap(b)$, is the probability distribution over $\R$ with probability mass function $z \mapsto \frac{1}{2b} \exp(-|z|/b)$.

Given a function $f: \cX^* \to \R$, its \emph{sensitivity} is defined as $\Delta(f) := \max_{D, D'} |f(D) - f(D')|$, where the maximum is over all pair $D, D'$ of adjacent datasets.

The Laplace mechanism~\cite{DworkMNS06} is an $\eps$-$\SDP$ mechanism that simply outputs $f(X) + \Lap(\Delta(f) / \eps)$.

\myparagraph{Group Privacy.}
The following fact is well-known and is often referred to as \emph{group  privacy}.

\begin{fact}[Group Privacy (e.g., \cite{Vadhan17})]\label{fact:group-dp}
Let $M: \cX^* \to \cY$ be an $(\eps, \delta)$-$\SDP$ mechanism and let $D, D' \in \N^{\cX}$ be such that $\|D - D'\| \leq t$, then, 
we have $M(D) \approx_{\eps', \delta'} M(D')$
where $\eps' = t\eps$ and $\delta' = \frac{e^{\eps'} - 1}{e^\eps - 1} \cdot \delta$.
\end{fact}

\myparagraph{DP Hyperparameter Tuning.}
We will also use the following result of Liu and Talwar~\cite{LT19} on DP hyperparameter tuning. We remark that some improvements in the constants has been made in \cite{PS21}, by using a different distribution of the number of repetitions. Nonetheless, since we are only interested in an asymptotic bound, we choose to work with the slightly simpler hyperparameter tuning algorithm from \cite{LT19}.

The hyperparameter tuning algorithm from \cite{LT19} allows us to take any DP ``base'' mechanism $\Mbase$, which outputs a candidate $y$ and a score $q \in \R$, run it multiple times and output a candidate with score that is below a certain threshold.\footnote{While DP Hyperparameter tuning is typically stated for choosing based on score {\em above} a threshold, the formulations are equivalent.} The precise description is in \Cref{alg:htune}.

\begin{algorithm}[ht]
\caption{DP Hyperparameter Tuning $\Mtuning$.\label{alg:htune}}
\begin{algorithmic}
\STATE {\bf Parameters:} Mechanism $\Mbase$, Threshold $s$, Number of Steps $T$, Stopping Probability $\gamma$.
\STATE {\bf Input:} Dataset $D$
\FOR{$j = 1, \dots, T$}
\STATE Let $(y, q) \gets \Mbase(D)$.
\IF{$q \leq s$}
\RETURN $y$ (and halt)
\ENDIF
\STATE With probability $\gamma$:
\STATE $\ \ $ {\bf return} $\perp$ (and halt)
\ENDFOR
\end{algorithmic}
\end{algorithm}

We will use the following DP guarantee of $\Mtuning$, which was shown in~\cite{LT19}\footnote{Note that this is a simplified version of \cite[Theorem 3.1]{LT19} where we simply set $\eps_0 = 1$.}.

\begin{theorem}[DP Hyperparameter Tuning~{\cite{LT19}}] \label{thm:dp-tuning}
For all $\eps > 0$, $\delta, \gamma \in [0, 1]$ and $T \ge 2/\gamma$,
if $\Mbase$ is $(\eps, \delta)$-$\SDP$,
then $\Mtuning$ (\Cref{alg:htune}) is $(2\eps + 1, 10e^{2\eps} \cdot \delta / \gamma)$.
\end{theorem}

\subsection{Weak Hardness}
\label{sec:sdp-each-block}

We start with a relatively weak hardness for the case of $\tau = 0$, i.e., the answer is considered correct iff it is the same as the input. To prove this, we recall a couple of facts.

The first is a simple relation between independent set and maximum matching. Let $\inds(G)$ denote the size of the maximum independent set of $G$.

\begin{fact}
\label{fac:matching}
For any graph $G = (V, E)$, there exists matching of size at least $(|V| - \inds(G)) / 2$.
\end{fact}

Let $\bbH^d$ denote the distance-$d$ graph on the hypercube, i.e., $\bbH^d = (\bit^n, E)$ where $(\bx, \bx') \in E$ iff $\|\bx - \bx'\|_1 \leq d$. 
Let $\binom{n}{\leq d} = \sum_{i = 0}^d \binom{n}{i}$.  The following standard lower bound follows from a ``packing argument''.

\begin{fact}
\label{fac:hcpacking}
For any $d \in \N$, $\inds(\bbH^{2d + 1}) \leq 2^n / \binom{n}{\leq d}$.
\end{fact}

We are now ready to prove a lower bound for the nearby problem.

\begin{theorem} \label{thm:each-block}
For any $\cR \subseteq \bit^n, d, \eps, \delta$, let $\eps' = (2d+1)\eps$ and $\delta' = \frac{e^{\eps'}-1}{e^\eps-1} \delta$. Then, for any $(\eps, \delta)$-$\SDP$ algorithm $M$,
 we have 
\begin{align*}
\sum_{x \in \cR} \Pr[M(x) \ne x] \geq 0.5 e^{-\eps'}(1 - \delta')\left(|\cR| - \frac{2^n}{\binom{n}{\leq d}}\right).
\end{align*}
\end{theorem}

\begin{proof}
Let $\bbH^{2d + 1}[\cR]$ denote the subgraph of $\bbH^{2d + 1}$ induced on $\cR$.
Notice that $\inds(\bbH^{2d + 1}[\cR]) \leq \inds(\bbH^{2d + 1})$. Therefore, by \Cref{fac:matching} and \Cref{fac:hcpacking}, we can conclude $\bbH^{2d + 1}[\cR]$ contains a matching of size at least $m \geq \left(|\cR| - 2^n / \binom{n}{\leq d}\right) / 2$.
Let the matching be $(x^1, \tx^1), \dots, (x^m, \tx^m)$.

For each $i \in [m]$, we have%
\begin{align*}
&\Pr[M(x^i) \ne x^i] + \Pr[M(\tx^i) \ne \tx^i] \\
&~\geq~ \Pr[M(x^i) = \tx^i] + \Pr[M(\tx^i) \ne \tx^i] \\
&~\geq~ e^{-\eps'}(\Pr[M(\tx^i) = \tx^i] - \delta') + \Pr[M(\tx^i) \ne \tx^i] \\
&~\geq~  e^{-\eps'}(\Pr[\cM(\tx^i) = \tx^i] + \Pr[\cM(\tx^i) \ne \tx^i] - \delta') \\
&~=~ e^{-\eps'}(1 - \delta').
\end{align*}
Adding this over all $i \in [m]$ yields the claimed bound.
\end{proof}

\subsection{Boosting the Distance}
\label{sec:boost-dist}

We can now prove a hardness for larger $\tau$ by dividing the coordinates into groups and applying the previously derived weak hardness result on each group. We note that the ``non-usefulness'' we get on the right hand side is still insufficient for \Cref{thm:no-dp-main}; this will be dealt with in \Cref{sec:boost}.

\begin{theorem} \label{thm:large-dist-err}
Let $n = n' \cdot b'$ for some $n', b' \in \N$.
For any $\cR \subseteq \bit^n, d, \eps, \delta, \zeta$, let $\eps' = (2d+1)\eps$ and $\delta' = \frac{e^{\eps'}-1}{e^\eps-1} \delta$. Then, for any $(\eps, \delta)$-$\SDP$ algorithm $M$,
 there exists $x \in \cR$ such that
\begin{align*}
& \Pr[u^{\NBP}_{\zeta \cdot b', \cR}(M(x), x) = 0]\\
&\textstyle~\geq\left(0.5 e^{-\eps'}(1 - \delta')\left(1 - \frac{2^n}{|\cR| \cdot \binom{n'}{\leq d}}\right)\right) - \zeta.
\end{align*}
\end{theorem}
\begin{proof}
Let $B_i := \{(i - 1)n' + 1, \dots, in'\}$ for all $i \in [b']$. Furthermore, let $\cR_{(B_i, z_{-B_i})}$ denote the set of all $x \in \cR$ such that $x_{-B_i} = z_{-B_i}$.

First, notice that
\begin{align*}
&\sum_{x \in \cR} \Pr[u^{\NBP}_{\zeta \cdot b', \cR}(M(x), x) = 0] \\
&\textstyle=~ \sum_{x \in \cR} \E_{y \gets M(x)} \indicator\set{\frac{|\{i \in [n] \mid y_i \ne x_i\}|}{b'} > \zeta} \\
&\textstyle~\geq~ \sum_{x \in \cR} \E_{y \gets M(x)} \indicator\set{\frac{|\{i \in [b'] \mid y_{B_i} \ne x_{B_i}\}|}{b'} > \zeta} \\
&\textstyle~\geq~ \sum_{x \in \cR} \E_{y \gets M(x)} \left[\Pr_{i \in [b']}[y_{B_i} \ne x_{B_i}] - \zeta\right]  \\
&\textstyle~=~ \left(\frac{1}{b'} \sum_{i \in [b']} \sum_{x \in \cR} \Pr[M(x)_{B_i} \ne x_{B_i}]\right) - \zeta |\cR| \\
&\textstyle~\geq~ \frac{1}{b'} \sum\limits_{\substack{i \in [b'] \\  z_{-B_i} \in \bit^{[n] \setminus B_i} \\ x \in \cR_{(B_i, z_{-B_i})}}}\Pr[M(x)_{B_i} \ne x_{B_i}] - \zeta |\cR|. 
\end{align*}
For each fixed $z_{-B_i} \in \bit^{[n] \setminus B_i}$, consider the mechanism $M': \bit^{B_i} \to \bit^{B_i}$ defined by $M'(x_{B_i}) := M_i(x_{B_i} \circ z_{-B_i})|_{B_i}$. It is clear that $M'$ is $(\eps, \delta)$-$\SDP$. Furthermore, observe that $\Pr[M(x)_{B_i} \ne x_{B_i}] = \Pr[M'(x) \ne x_{B_i}]$ for all $x \in \cR_{(B_i, z_{-B_i})}$. Therefore, by applying \Cref{thm:each-block} and plugging it back into the above, we get
\begin{align*}
&\sum_{x \in \cR} \Pr[u^{\NBP}_{\zeta \cdot b', \cR}(M(x), x) = 0] \\
&\textstyle~\geq~ \frac{1}{b'} \sum\limits_{\substack{i \in [b'] \\ z_{-B_i}}} 0.5 e^{-\eps'}(1 - \delta')\left(|\cR_{(B_i, z_{-B_i})}| - \frac{2^{n'}}{\binom{n'}{\leq d}}\right) \\
& \qquad - \zeta |\cR| \\
&\textstyle~=~ \frac{1}{b'} \sum_{i \in [b']} 0.5 e^{-\eps'}(1 - \delta')\left(|\cR| - \frac{2^{n - n'} \cdot 2^{n'}}{\binom{n'}{\leq d}}\right)\\
& \qquad - \zeta |\cR| \\
&\textstyle~=~ \left(0.5 e^{-\eps'}(1 - \delta')\left(|\cR| - \frac{2^n}{\binom{n'}{\leq d}}\right)\right) - \zeta |\cR|.
\end{align*}
Dividing by $|\cR|$ then gives us the claimed bound.
\end{proof}

\subsection{Boosting the Failure Probability}
\label{sec:boost}

We will now prove the last part of the lower bound, which is to show that the existence of even slightly useful mechanism also leads to an existence of a highly useful mechanism, albeit at a slight increase in the distance threshold. The formal statement and its proof are given below; the proof uses the DP hyperparameter tuning algorithm (\Cref{thm:dp-tuning}).

\begin{theorem} \label{thm:boosting-usefulness}
Suppose that there exists an $(\eps, \delta)$-$\SDP$ mechanism $M: \bit^n \to \bit^n$ that is $\alpha$-useful for $u^{\NBP}_{\tau, \cR}$. Then, for all $C > 0$, there exists an $(\eps', \delta')$-$\SDP$ mechanism $\hM: \bit^n \to \bit^n$ that is $(1 - 1/n^{C})$-useful for $u^{\NBP}_{\tau', \cR}$ where $\eps' = 4\eps + 1, \delta' = O\left(\frac{e^{4\eps} n^C \ln n}{\alpha} \cdot \delta\right)$ and $\tau' = \tau + O\left(\frac{\ln n}{\alpha}\right)$.
\end{theorem}

\begin{proof}
First, let us construct the mechanism $\Mbase: \bit^n \to \bit^n \times \R$ as follows:
\begin{itemize}[leftmargin=*]
\item On input $x \in \bit^n$, first let $y \leftarrow M(x)$.
\item Then, let $q = \|x - y\|_1 + z$ where $z \sim \Lap(1/\eps)$.
\item Output $(y, q)$.
\end{itemize}
Since $M$ is $(\eps, \delta)$-$\SDP$ and the Laplace mechanism is $\eps$-$\SDP$, the basic composition theorem implies that the entire $\Mbase$ mechanism is $(2\eps, \delta)$-$\SDP$.

Let $\hT = \ln(5n^{C})/\alpha$. Let $\tau' = \tau + 2\log(10n^{C}\hT) / \eps$.
We now apply \Cref{alg:htune} with $\gamma = 0.5 / (n^{C} \hT), T = 2 / \gamma$ and threshold $s = \tau' - \log(10n^{C}\hT) / \eps$. \Cref{thm:dp-tuning} ensures that the resulting algorithm $\Mtuning$ is $(4\eps + 1, 10e^{4\eps} \delta / \gamma)$-$\SDP$. Our final mechanism $\hM$ is the mechanism that runs $\Mtuning$. If the output is not $\perp$, $\hM$ returns that output. Otherwise, $\hM$ returns an arbitrary element of $\bit^n$. Since $\hM$ is a post-processing of $\Mtuning$, we have $\hM$ is also $(4\eps + 1, 10e^{4\eps} \delta / \gamma)$-$\SDP$.

We will next show that $\Mtuning$ is $(1 - 1/n^{C})$-useful for $u^{\NBP}_{\tau', \cR}$. By definition of the utility function, this immediately holds for any $x \notin \cR$. Therefore, we may only consider any $x \in \cR$. Consider $\Mtuning$ on such an $x$. Let $y^i, z^i, q^i$ denote the corresponding values of $y, z, q$ in the $i$-th run of $\Mbase$.
We consider the following three events:
\begin{itemize}[leftmargin=*]
\item Let $\mathcal{E}_1$ denote the event that $|\|x^i - y^i\|_1 - q^i| > \log(10n^{C}\hT) / \eps$ for some $i \in [\hT]$.
\item Let $\mathcal{E}_2$ denote the event that $u_{\tau, \cR}(y_i) = 0$ for all $i \in [\hT]$. 
\item Let $\mathcal{E}_3$ denote the event that $\Mtuning$ halts and returns $\bot$ in the first $\hT$ steps.
\end{itemize}
Before we bound the probability of each event, notice that, if none of $\mathcal{E}_1, \mathcal{E}_2, \mathcal{E}_3$ occurs, we must have $u^{\NBP}_{\tau', \cR}(y) = 1$ (where $y$ denotes the output of $\hM$), since $s - \tau, \tau' - s \geq \log(10n^{C}\hT) / \eps$.
That is,
\begin{align*}
\Pr_{y \leftarrow \hM(x)}[u^{\NBP}_{\tau', \cR}(y) = 0] & \leq \Pr[\mathcal{E}_1 \vee \mathcal{E}_2 \vee \mathcal{E}_3] \\
& \leq \Pr[\mathcal{E}_1] + \Pr[\mathcal{E}_2] + \Pr[\mathcal{E}_3].
\end{align*}

We will now bound the probability for each event. For $\mathcal{E}_1$, it immediately follows from the Laplace tail bound together with a union bound that
\begin{align*}
\Pr[\mathcal{E}_1] \leq \hT \cdot 2/(10n^{C}\hT) = 0.2/n^{C}.
\end{align*}
For $\mathcal{E}_2$, the $\alpha$-usefulness of $M$ implies that
\begin{align*}
\Pr[\mathcal{E}_2] \leq (1 - \alpha)^{\hT} \leq 0.2/n^{C}.
\end{align*}
Finally, for $\mathcal{E}_3$, a simple union bound gives
\begin{align*}
\Pr[\mathcal{E}_3] \leq \gamma \cdot \hT \leq 0.5/n^{C}.
\end{align*}
By combining the four inequalities above, we have
\begin{align*}
\Pr_{y \leftarrow \hM(x)}[u^{\NBP}_{\tau', \cR}(y) = 0] < 1/n^{C},
\end{align*}
as desired.
\end{proof}

\subsection{Putting Things Together: Proof of \Cref{thm:no-dp-main}}

\begin{proof}[Proof of \Cref{thm:no-dp-main}]
Suppose for the sake of contradiction that, for some constant $\beps > 0$ and $\bdelta = 1/n^{-27}$ there exists an $(\beps, \bdelta)$-$\SDP$ mechanism $\bM$ that is $0.01$-useful for $u^{\NBP}_{\tau, \cR}$ for every $n \in \N$; recall $\tau_n \le n^{0.9}$.

Using \Cref{thm:boosting-usefulness} with $C=26$, there is a $(\hat{\beps}, \hat{\bdelta})$ mechanism $M'_n$ for $\hat{\beps} = 4\beps + 1$ and $\hat{\bdelta} = O(n^{26} \log n \cdot \bdelta) = O(\log n / n)$ that is $(1 - 1/n^{26})$-useful for $u^{\NBP}_{\tau'_n, \cR_n}$ where $\tau'_n = \tau_n + O(\log n) = O(n^{0.9})$. Plugging this into \Cref{thm:large-dist-err} with $\cR = \cR_n, n' = n^{0.05}, b' = n^{0.95}, \zeta = \tau'_n / b' \leq O(n^{-0.05}), \beps = \hat{\beps}, \bdelta = \hat{\bdelta}, d = (\log n^{0.04})/3\eps$ (which gives $\beps' \leq \log(2n^{0.04})$ and $\bdelta' = O(\log n / n^{0.96}) = o_n(1)$ in \Cref{thm:large-dist-err}), we have
\begin{align*}
\frac{1}{n^{26}} \geq &~  \left(0.5 \cdot e^{-\log(2n^{0.04})}(1 - o_n(1))\left(1 - o_n(1)\right)\right)  \\
&  ~\quad - O(n^{-0.05}) \\
= & ~ O(n^{-0.04}) \cdot (1 - o(1)) - O(n^{-0.05}),
\end{align*} 
which is a contradiction for any sufficiently large $n$.

 \end{proof}

\section{Putting Things Together: Proof of \texorpdfstring{\Cref{thm:main-formal}}{Theorem~\ref{thm:main-formal}}}
\label{sec:proof-main-thm}

Our main theorem follows from combining the main results from the previous two sections.

\begin{proof}[Proof of \Cref{thm:main-formal}]
Let $u = u^{\VLDS}_{\tau, \cR, V}$ be as given in \Cref{thm:cdp-main}, which immediately yields the existence of an $\beps_{\CDP}$-$\CDP$ mechanism that is $(1 - o(1))$-useful.
Furthermore, since $|\cR| \geq 2^{n} / n^{o(\log n)}$, \Cref{thm:no-dp-main} implies that for any constant $\beps_{\SDP}, \balpha > 0$, there is $\bdelta_{\SDP} = 1/n^{27}$ such that no $(\beps_{\SDP}, \bdelta_{\SDP})$-$\SDP$ mechanism is $\balpha$-useful for $u^{\NBP}_{\tau,\cR}$. Finally, applying \Cref{lem:low-diameter-to-nearbypoint}, we can conclude that no $(\beps_{\SDP}, \bdelta_{\SDP})$-$\SDP$ mechanism is $\balpha$-useful for $u^{\VLDS}_{\tau, \cR, V}$. This concludes our proof.
\end{proof}%
\section{Conclusion and Discussion}
\label{sec:conclusion}

In this work, we give a first task that, under certain assumptions, admits an efficient CDP algorithm but does not admit an SDP algorithm (even inefficient ones). As mentioned in \Cref{sec:intro}, perhaps the most intriguing next direction would be to see if there are more ``natural'' tasks for which $\CDP$ algorithms can go beyond known $\SDP$ lower bounds. 

On the technical front, there are also a few interesting directions. For example, it would be interesting to see if the three assumptions in our paper can be removed, relaxed, or replaced (by perhaps more widely believed assumptions). Alternatively, we can ask the opposite question: what are the (cryptographic) assumptions necessary for separating $\CDP$ and $\SDP$? Such a question has been extensively studied in the multiparty model~\cite{HaitnerMST22,GoyalMPS13,GoyalKMPS16,HaitnerMSS19,HaitnerNOSS18}; for example, it is known that key-agreement is necessary and sufficient to get better-than-local-DP protocol for inner product in the two-party setting~\cite{HaitnerMST22}. Achieving such a result in our setting would significantly deepen our understanding of the $\CDP$-vs-$\SDP$ question in the central model.

Another possible improvement is to strengthen the hardness of the adversary. In this paper, we only consider polynomial-time adversaries. Indeed, our $\CDP$ mechanism does \emph{not} remain $\CDP$ against quasi-polynomial adversary. The reason is that we choose the hash value length to be only $o(\log^2 \lambda)$ in \Cref{as:crklhf}, so a trivial ``guess-and-check'' algorithm can break this assumption in time $\lambda^{O(\log \lambda)}$. However, as far as we are aware, there is no inherent barrier in proving a separation with $\CDP$ that holds even against, e.g., sub-exponential time adversaries. Achieving such a result (potentially under stronger or different assumptions) would definitely be interesting.

Furthermore, our task (or more precisely the utility function) is non-uniform (through the choice of $\upsilon_n$). It would also be interesting to have a uniform task.%

\subsection*{Acknowledgments}
We thank Prabhanjan Ananth for helpful discussions about differing-inputs obfuscation, and anonymous reviewers for helpful comments.

\bibliographystyle{alpha}
\bibliography{main.bbl}

\newcommand{\etalchar}[1]{$^{#1}$}
\begin{thebibliography}{GGHW17}

\bibitem[ABG{\etalchar{+}}13]{AnanthBGSZ13}
Prabhanjan Ananth, Dan Boneh, Sanjam Garg, Amit Sahai, and Mark Zhandry.
\newblock Differing-inputs obfuscation and applications.
\newblock {\em {IACR} Cryptol. ePrint Arch.}, page 689, 2013.

\bibitem[Abo18]{abowd2018us}
John~M Abowd.
\newblock The {US Census Bureau} adopts differential privacy.
\newblock In {\em KDD}, pages 2867--2867, 2018.

\bibitem[{App}17]{dp2017learning}
{Apple Differential Privacy Team}.
\newblock Learning with privacy at scale.
\newblock {\em Apple Machine Learning Journal}, 2017.

\bibitem[BCP14]{BoyleCP14}
Elette Boyle, Kai{-}Min Chung, and Rafael Pass.
\newblock On extractability obfuscation.
\newblock In {\em TCC}, pages 52--73, 2014.

\bibitem[BCV16]{BunCV16}
Mark Bun, Yi{-}Hsiu Chen, and Salil~P. Vadhan.
\newblock Separating computational and statistical differential privacy in the
  client-server model.
\newblock In {\em TCC}, pages 607--634, 2016.

\bibitem[BGI{\etalchar{+}}01]{BarakGIRSVY01}
Boaz Barak, Oded Goldreich, Russell Impagliazzo, Steven Rudich, Amit Sahai,
  Salil~P. Vadhan, and Ke~Yang.
\newblock On the (im)possibility of obfuscating programs.
\newblock In {\em CRYPTO}, pages 1--18, 2001.

\bibitem[BGI{\etalchar{+}}12]{BarakGIRSVY12}
Boaz Barak, Oded Goldreich, Russell Impagliazzo, Steven Rudich, Amit Sahai,
  Salil~P. Vadhan, and Ke~Yang.
\newblock On the (im)possibility of obfuscating programs.
\newblock {\em J. {ACM}}, 59(2):6:1--6:48, 2012.

\bibitem[BKP18]{BitanskyKP18}
Nir Bitansky, Yael~Tauman Kalai, and Omer Paneth.
\newblock Multi-collision resistance: a paradigm for keyless hash functions.
\newblock In {\em STOC}, pages 671--684, 2018.

\bibitem[BNO08]{BeimelNO08}
Amos Beimel, Kobbi Nissim, and Eran Omri.
\newblock Distributed private data analysis: Simultaneously solving how and
  what.
\newblock In {\em CRYPTO}, pages 451--468, 2008.

\bibitem[BOV07]{BarakOV07}
Boaz Barak, Shien~Jin Ong, and Salil~P. Vadhan.
\newblock Derandomization in cryptography.
\newblock {\em {SIAM} J. Comput.}, 37(2):380--400, 2007.

\bibitem[BP15a]{BitanskyP15}
Nir Bitansky and Omer Paneth.
\newblock Zaps and non-interactive witness indistinguishability from
  indistinguishability obfuscation.
\newblock In {\em TCC}, pages 401--427, 2015.

\bibitem[BP15b]{BoyleP15}
Elette Boyle and Rafael Pass.
\newblock Limits of extractability assumptions with distributional auxiliary
  input.
\newblock In {\em ASIACRYPT}, pages 236--261, 2015.

\bibitem[BR93]{BellareR93}
Mihir Bellare and Phillip Rogaway.
\newblock Ccs.
\newblock pages 62--73, 1993.

\bibitem[BSW16]{BellareSW16}
Mihir Bellare, Igors Stepanovs, and Brent Waters.
\newblock New negative results on differing-inputs obfuscation.
\newblock In {\em EUROCRYPT}, pages 792--821, 2016.

\bibitem[CGH04]{CanettiGH04}
Ran Canetti, Oded Goldreich, and Shai Halevi.
\newblock The random oracle methodology, revisited.
\newblock {\em J. {ACM}}, 51(4):557--594, 2004.

\bibitem[De12]{De12}
Anindya De.
\newblock Lower bounds in differential privacy.
\newblock In {\em TCC}, pages 321--338, 2012.

\bibitem[DKM{\etalchar{+}}06]{DworkKMMN06}
Cynthia Dwork, Krishnaram Kenthapadi, Frank McSherry, Ilya Mironov, and Moni
  Naor.
\newblock Our data, ourselves: Privacy via distributed noise generation.
\newblock In {\em EUROCRYPT}, pages 486--503, 2006.

\bibitem[DKY17]{ding2017collecting}
Bolin Ding, Janardhan Kulkarni, and Sergey Yekhanin.
\newblock Collecting telemetry data privately.
\newblock In {\em NeurIPS}, pages 3571--3580, 2017.

\bibitem[DMNS06]{DworkMNS06}
Cynthia Dwork, Frank McSherry, Kobbi Nissim, and Adam~D. Smith.
\newblock Calibrating noise to sensitivity in private data analysis.
\newblock In {\em TCC}, pages 265--284, 2006.

\bibitem[DMT07]{DworkMT07}
Cynthia Dwork, Frank McSherry, and Kunal Talwar.
\newblock The price of privacy and the limits of {LP} decoding.
\newblock In {\em STOC}, pages 85--94, 2007.

\bibitem[DN03]{DinurN03}
Irit Dinur and Kobbi Nissim.
\newblock Revealing information while preserving privacy.
\newblock In {\em PODS}, pages 202--210, 2003.

\bibitem[EPK14]{erlingsson2014rappor}
{\'U}lfar Erlingsson, Vasyl Pihur, and Aleksandra Korolova.
\newblock {RAPPOR}: Randomized aggregatable privacy-preserving ordinal
  response.
\newblock In {\em CCS}, pages 1054--1067, 2014.

\bibitem[GGHW17]{GargGHW17}
Sanjam Garg, Craig Gentry, Shai Halevi, and Daniel Wichs.
\newblock On the implausibility of differing-inputs obfuscation and extractable
  witness encryption with auxiliary input.
\newblock {\em Algorithmica}, 79(4):1353--1373, 2017.

\bibitem[GKM{\etalchar{+}}16]{GoyalKMPS16}
Vipul Goyal, Dakshita Khurana, Ilya Mironov, Omkant Pandey, and Amit Sahai.
\newblock Do distributed differentially-private protocols require oblivious
  transfer?
\newblock In {\em ICALP}, pages 29:1--29:15, 2016.

\bibitem[GKY11]{GroceKY11}
Adam Groce, Jonathan Katz, and Arkady Yerukhimovich.
\newblock Limits of computational differential privacy in the client/server
  setting.
\newblock In {\em TCC}, pages 417--431, 2011.

\bibitem[GMPS13]{GoyalMPS13}
Vipul Goyal, Ilya Mironov, Omkant Pandey, and Amit Sahai.
\newblock Accuracy-privacy tradeoffs for two-party differentially private
  protocols.
\newblock In {\em CRYPTO}, pages 298--315, 2013.

\bibitem[GOS12]{GrothOS12}
Jens Groth, Rafail Ostrovsky, and Amit Sahai.
\newblock New techniques for noninteractive zero-knowledge.
\newblock {\em J. {ACM}}, 59(3):11:1--11:35, 2012.

\bibitem[Gre16]{greenberg2016apple}
Andy Greenberg.
\newblock {Apple's} ``differential privacy'' is about collecting your data --
  but not your data.
\newblock {\em Wired, June}, 13, 2016.

\bibitem[GT08]{GreenT08}
Ben Green and Terence Tao.
\newblock The primes contain arbitrarily long arithmetic progressions.
\newblock {\em Annals of Mathematics}, 167(2):481--547, 2008.

\bibitem[HMSS19]{HaitnerMSS19}
Iftach Haitner, Noam Mazor, Ronen Shaltiel, and Jad Silbak.
\newblock Channels of small log-ratio leakage and characterization of two-party
  differentially private computation.
\newblock In {\em TCC}, pages 531--560, 2019.

\bibitem[HMST22]{HaitnerMST22}
Iftach Haitner, Noam Mazor, Jad Silbak, and Eliad Tsfadia.
\newblock On the complexity of two-party differential privacy.
\newblock In {\em STOC}, pages 1392--1405, 2022.

\bibitem[HNO{\etalchar{+}}18]{HaitnerNOSS18}
Iftach Haitner, Kobbi Nissim, Eran Omri, Ronen Shaltiel, and Jad Silbak.
\newblock Computational two-party correlation: {A} dichotomy for key-agreement
  protocols.
\newblock In {\em FOCS}, pages 136--147, 2018.

\bibitem[HT10]{HardtT10}
Moritz Hardt and Kunal Talwar.
\newblock On the geometry of differential privacy.
\newblock In {\em STOC}, pages 705--714, 2010.

\bibitem[IPS15]{IshaiPS15}
Yuval Ishai, Omkant Pandey, and Amit Sahai.
\newblock Public-coin differing-inputs obfuscation and its applications.
\newblock In {\em TCC}, pages 668--697, 2015.

\bibitem[JLS21]{JainLS21}
Aayush Jain, Huijia Lin, and Amit Sahai.
\newblock Indistinguishability obfuscation from well-founded assumptions.
\newblock In {\em STOC}, pages 60--73, 2021.

\bibitem[KT18]{LinkedINDP1}
Krishnaram Kenthapadi and Thanh T.~L. Tran.
\newblock {PriPeARL}: A framework for privacy-preserving analytics and
  reporting at {LinkedIn}.
\newblock In {\em CIKM}, pages 2183--2191, 2018.

\bibitem[LT19]{LT19}
Jingcheng Liu and Kunal Talwar.
\newblock Private selection from private candidates.
\newblock In {\em STOC}, pages 298--309, 2019.

\bibitem[Man22]{Manurangsi22}
Pasin Manurangsi.
\newblock Tight bounds for differentially private anonymized histograms.
\newblock In {\em SOSA}, pages 203--213, 2022.

\bibitem[MMP{\etalchar{+}}10]{McGregorMPRTV10}
Andrew McGregor, Ilya Mironov, Toniann Pitassi, Omer Reingold, Kunal Talwar,
  and Salil~P. Vadhan.
\newblock The limits of two-party differential privacy.
\newblock In {\em FOCS}, pages 81--90, 2010.

\bibitem[MPRV09]{MironovPRV09}
Ilya Mironov, Omkant Pandey, Omer Reingold, and Salil~P. Vadhan.
\newblock Computational differential privacy.
\newblock In {\em CRYPTO}, pages 126--142, 2009.

\bibitem[PS21]{PS21}
Nicolas Papernot and Thomas Steinke.
\newblock Hyperparameter tuning with renyi differential privacy.
\newblock {\em CoRR}, abs/2110.03620, 2021.

\bibitem[RSP{\etalchar{+}}21]{LinkedInDP2}
Ryan Rogers, Subbu Subramaniam, Sean Peng, David Durfee, Seunghyun Lee,
  Santosh~Kumar Kancha, Shraddha Sahay, and Parvez Ahammad.
\newblock {LinkedIn’s} audience engagements {API}: A privacy preserving data
  analytics system at scale.
\newblock {\em J. Priv. Confiden.}, 11(3), 2021.

\bibitem[RTTV08]{ReingoldTTV08}
Omer Reingold, Luca Trevisan, Madhur Tulsiani, and Salil~P. Vadhan.
\newblock Dense subsets of pseudorandom sets.
\newblock In {\em FOCS}, pages 76--85, 2008.

\bibitem[Sha14]{CNET2014Google}
Stephen Shankland.
\newblock How {Google} tricks itself to protect {Chrome} user privacy.
\newblock {\em CNET, October}, 2014.

\bibitem[Sur19]{Suresh19}
Ananda~Theertha Suresh.
\newblock Differentially private anonymized histograms.
\newblock In {\em NeurIPS}, pages 7969--7979, 2019.

\bibitem[TZ08]{TaoZ08}
Terence Tao and Tamar Ziegler.
\newblock {The primes contain arbitrarily long polynomial progressions}.
\newblock {\em Acta Mathematica}, 201(2):213 -- 305, 2008.

\bibitem[Vad17]{Vadhan17}
Salil~P. Vadhan.
\newblock The complexity of differential privacy.
\newblock In {\em Tutorials on the Foundations of Cryptography}, pages
  347--450. Springer International Publishing, 2017.

\bibitem[War65]{warner1965randomized}
Stanley~L Warner.
\newblock Randomized response: A survey technique for eliminating evasive
  answer bias.
\newblock {\em JASA}, 60(309):63--69, 1965.

\end{thebibliography}

\appendices
\section{Comparison of various \texorpdfstring{$\diO$}{diO} assumptions}\label{apx:diO-comparison}

We review and compare the various notions of differing inputs obfuscation, showing that the notion of $\diOpcS$ (\Cref{def:diO}) is in fact weaker (or at least, no stronger) than all notions of differing inputs obfuscation studied in literature.

The definition of $\diO$ as given by \cite{BarakGIRSVY12} did not include the notion of a sampler. Informally speaking, it requires that for all efficient adversaries $A$ there is an efficient adversary $A'$ such that if $A$ can distinguish the obfuscation of a circuit $C_0$ from the obfuscation of $C_1$, then for any circuits $C_0'$ and $C_1'$ that are functionally equivalent to $C_0$ and $C_1$ respectively, $A'$ can find an input on which $C_0$ and $C_1$ disagree.

This notion is stronger than the corresponding notion involving samplers. Since most applications of differing-inputs obfuscation in literature are stated using differing-inputs samplers, we will only refer to $\diO$ notions that involve these.

\begin{definition}[Differing-Inputs Circuit Sampler~\cite{AnanthBGSZ13}]\label{def:genSamplers}
An efficient non-uniform sampling algorithm $\Sampler = \set{\Sampler_n}$ is a \emph{differing-inputs sampler} for the parameterized collection $\cC = \set{\cC_n}$ of circuits if the output of $\Sampler_n$ is distributed over $\cC_n \times \cC_n \times \bit^*$ and for every efficient non-uniform algorithm $\cA = \set{\cA_n}$, there exists a negligible function $\negl(\cdot)$ such that for all $n \in \N$:
\[
\Pr_\theta\begin{bmatrix*}[l]
    C_0(y) \ne C_1(y) :\\
    (C_0, C_1, \aux) \gets \Sampler_n(\theta),\\
    y \gets \cA_n(C_0, C_1, \aux)
    \end{bmatrix*}
    ~\le~ \negl(n).
\]
\end{definition}

\begin{description}[leftmargin=2mm]
\item[{\boldmath Plain Sampler.}] We call a differing-inputs sampler as a {\em Plain Sampler} if $\aux$ is always $\bot$.
\item [{\boldmath Public-Coin Sampler.}] We call a differing-inputs sampler as {\em Public-Coin Sampler} if $\aux$ is equal to $\theta$ (precisely \Cref{def:pcSamplers}).
\item[{\boldmath General Sampler.}] We call a differing-inputs sampler as a {\em General Sampler} whenever we want to emphasize that $\aux$ is allowed to be any function of $\theta$. In particular, plain and public-coin samplers are special cases of general samplers.
\end{description}
Note that, the more information that $\aux$ is allowed to contain, the more restricted the distribution over circuit pairs $(C_0, C_1)$ gets. In particular, any public-coin Sampler remains a differing-inputs Sampler if we set $\aux$ to be some function of $\theta$ (instead of being all of $\theta$), and similarly, any general differing-inputs Sampler can be converted to a plain-Sampler by simply setting $\aux=\bot$.

We can consider two notions of security of differing inputs obfuscators, depending on whether or not the distinguisher has access to $\aux$. Recall that the ``differing-inputs'' condition in \Cref{def:diO} was
\begin{align}
\begin{vmatrix*}[l]
\Pr_{\theta} \left [D_n(\diO(1^n, C_0)) = 1 \right] \\
- \Pr_{\theta} \left [D_n(\diO(1^n, C_1)) = 1 \right ]
\end{vmatrix*} &~\le~ \negl(n).\label{eq:no-aux}
\end{align}
On the other hand, we could consider a different notion where for any general sampler $\Sampler$, for $(C_0, C_1, \aux) \gets \Sampler_n(\theta)$, we replace the ``differing-inputs'' condition with
\begin{align}
\begin{vmatrix*}[l]
\Pr_{\theta} \left [D_n(\diO(1^n, C_0), \aux) = 1 \right] \\
- \Pr_{\theta} \left [D_n(\diO(1^n, C_1), \aux) = 1 \right ]
\end{vmatrix*} &~\le~ \negl(n).\label{eq:with-aux}
\end{align}

Depending on the type of sampler (plain or public-coin or general) and the notion of security for differing inputs obfuscators (\eqref{eq:no-aux} or \eqref{eq:with-aux}), we get various kinds of $\diO$ assumptions, which we list below.

\begin{description}[leftmargin=2mm]
\item[{\boldmath Plain $\diO$.}]
We refer to $\plaindiO$ as the notion of $\diO$ that holds only against plain samplers. Note, there is no difference here between the security notions of \eqref{eq:no-aux} and \eqref{eq:with-aux}, since $\aux = \bot$ anyway.

\item[{\boldmath Public-Coin $\diO$.}]
We refer to $\pcdiO$, as the notion of public-coin $\diO$ defined by \cite{IshaiPS15}, corresponding to the notion of $\diO$ that holds only against public-coin samplers, where the distinguisher also has access to $\aux = \theta$, as in \eqref{eq:with-aux}.

\item[{\boldmath General $\diO$.}]
We refer to $\gendiO$, as the notion of general $\diO$ defined by \cite{AnanthBGSZ13}, corresponding to the notion of $\diO$ that holds for general samplers, and where the distinguisher also has access to $\aux$, as in \eqref{eq:with-aux}.

\item[{\boldmath $\diO$ for General Samplers.}]
We define $\diOgenS$ as the notion of $\diO$ that holds only against general samplers, but where the distinguisher does not have access to $\aux$, as in \eqref{eq:no-aux}.

\item[{\boldmath $\diO$ for Public-Coin Samplers.}]
This is precisely \Cref{def:diO}, where the security of $\diO$ holds only for public-coin samplers, where the distinguisher does not have access to $\aux = \theta$, as in \eqref{eq:no-aux}.
\end{description}

\begin{figure}[t]
\centering
\begin{tikzpicture}
	\def\scle{1.7}
	\node (plaindio) at (-1*\scle,0*\scle) {$\plaindiO$};
	\node (pcdio) at (1*\scle,0*\scle) {$\pcdiO$};
	\node (gendio) at (0*\scle,1*\scle) {$\gendiO$};
	\node (diogen) at (-1*\scle,-1*\scle) {$\diOgenS$};
	\node (diopcS) at (0*\scle,-2*\scle) {$\diOpcS$};
	\path[{Stealth[length=2mm,width=2mm]}-, line width=1pt]
	(plaindio) edge (gendio)
	(pcdio) edge (gendio)
	(diopcS) edge (diogen)
	(diopcS) edge (pcdio)
        (diogen) edge[bend left=30] (plaindio)
        (plaindio) edge[bend left=30] (diogen);
\end{tikzpicture}
\caption{Comparisons between different $\diO$ assumptions, where $\mathsf{A} \to \mathsf{B}$ denotes that existence of $\mathsf{A}$ implies existence of $\mathsf{B}$, or in other words, existence of $\mathsf{A}$ is a {\em stronger} assumption than existence of $\mathsf{B}$. Existence of $\diOpcS$ (assumption used in this paper) is the weakest among all the notions.}
\label{fig:diO-ass-comp}
\end{figure}

\myparagraph{Comparison between different $\diO$ assumptions.}
The comparison between the assumptions asserting existence of each type of $\diO$ is illustrated in \Cref{fig:diO-ass-comp}, with justification for each arrow given as follows:
\begin{itemize}[leftmargin=*]
	\item Existence of $\gendiO$ implies existence of $\plaindiO$ and $\pcdiO$, since both are special cases corresponding to plain samplers and public-coin samplers respectively.
	\item To the best of knowledge, it is unknown whether the assumptions of existence of $\plaindiO$ and the existence of $\pcdiO$ are comparable or not.
	\item Existence of $\plaindiO$ implies existence of $\diOgenS$ since any general sampler can be converted to a plain sampler by simply setting $\aux=\bot$; note that the distinguisher (in the definition of $\diO$) does not have access to $\aux$ in either case.
	\item Existence of $\diOgenS$ implies existence of $\plaindiO$ and $\diOpcS$ since both are special cases corresponding to plain samplers and public-coin samplers respectively.
	\item Existence of $\pcdiO$ implies existence of $\diOpcS$, since the distinguisher in the definition of $\diOpcS$ does not have access to $\theta$ and hence is less powerful.
\end{itemize}

\noindent Finally, one may wonder, what was special about the application of $\diO$ in this paper that only required $\diOpcS$ and not $\gendiO$ or $\pcdiO$ as in prior work in cryptography. The main reason is that, in cryptographic applications, an $\aux$ is provided to adversaries to enable certain cryptographic functionality (such as by revealing some public key parameters), and thus, it is required that the $\diO$ is secure even given knowledge of this $\aux$ information. In applications of $\pcdiO$, the distinguisher typically does not have access to all of $\theta$ (such as some secret key parameters may be hidden), but security given knowledge of entire $\theta$ implies security given partial knowledge of $\theta$. In the setting of this paper, there wasn't any particular functionality that needed to be enabled, other than basic circuit evaluation, and the particular circuit samplers of interest were public-coin differing inputs samplers, which is why it suffices to only assume $\diOpcS$.%
\end{document}